%% file: subsetsumapx.tex
\documentclass[11pt,english,numberwithinsect]{article}

\usepackage[left=1in,right=1in,top=1in,bottom=1in]{geometry}

\usepackage{amsthm}
\usepackage{amsmath}
\usepackage{amssymb}

\usepackage[noend]{algpseudocode}
\usepackage{algorithm}

\usepackage{graphicx}
\usepackage{calc}

\usepackage{tikz}

\usepackage[numbers,sort&compress]{natbib}

\usepackage{subcaption}

\usepackage{color}

\usepackage{wrapfig}

\usepackage{xspace}

\usepackage[bookmarks]{hyperref}

\usepackage{varioref} 

\usepackage{footnote} 

\usepackage{enumitem} 
\setlist{leftmargin=5.5mm}

\mathchardef\mhyphen="2D

\newcommand{\eps}{\ensuremath{\varepsilon}}

\newcommand{\cS}{{\cal S}}
\newcommand{\Ssp}{{\cal S}_{\Delta \text{-} {\rm{sp}}}}

\newtheorem{thm}{Theorem}[section]
\newtheorem{lem}[thm]{Lemma}
\newtheorem{rem}[thm]{Remark}
\newtheorem{obs}[thm]{Observation}
\newtheorem{cor}[thm]{Corollary}
\newtheorem{defn}[thm]{
Definition}

\newtheorem{claim}[thm]{Claim}

\global\long\def\Oh{{\cal O}}
\global\long\def\tOh{\widetilde{\Oh}}

\newcommand{\pl}[1][]{+_{#1}}
\newcommand{\plt}{\pl[t]}
\newcommand{\N}{\mathbb{N}}
\newcommand{\Z}{\mathbb{Z}}
\newcommand{\Ex}{\mathbb{E}}

\newcommand{\poly}{\textup{poly}}
\newcommand{\polylog}{\textup{polylog}}

\newcommand{\ApxSubsetSum}{\textsc{ApxSubsetSum}\xspace}
\newcommand{\DisSubsetSum}{\textsc{GapSubsetSum}\xspace}
\newcommand{\SubsetSum}{\textsc{SubsetSum}\xspace}

\newcommand{\Partition}{\textsc{Partition}\xspace}

\newcommand{\minconv}{\textsc{MinConv}\xspace}
\newcommand{\Tminconv}{T_{\minconv}}
\newcommand{\maxconv}{\textsc{MaxConv}\xspace}
\newcommand{\Knapsack}{\textsc{Knapsack}\xspace}

\newcommand{\apx}{approximates\xspace}

\newcommand{\Dtapx}{$(t,\Delta)$-approximates\xspace}
\newcommand{\wDtapx}{w.h.p.\ $(t,\Delta)$-approximates\xspace}
\newcommand{\Dinftyapx}{$(\infty,\Delta)$-approximates\xspace}
\newcommand{\sDinftyapx}{sparsely $(\infty,\Delta)$-approximates\xspace}
\newcommand{\sDtapx}{sparsely $(t,\Delta)$-approximates\xspace}
\newcommand{\wsDtapx}{w.h.p.\ sparsely $(t,\Delta)$-approximates\xspace}
\newcommand{\Dtprimeapx}{$(t',\Delta)$-approximates\xspace}

\newcommand{\Apxp}{\textup{apx}_t^+}
\newcommand{\Apxpprime}{\textup{apx}_{t'}^+}
\newcommand{\Apxm}{\textup{apx}_t^-}
\newcommand{\Apxmprime}{\textup{apx}_{t'}^-}
\newcommand{\Dsparse}{$\Delta$-sparse\xspace}

\newcommand{\OPT}{\textup{OPT}}

\begin{document}

\title{A Fine-Grained Perspective on \\ Approximating Subset Sum and Partition\footnote{This work is part of the project TIPEA that has received funding from the European
Research Council (ERC) under the European Unions Horizon 2020 research and innovation programme (grant
agreement No.\ 850979).}}
\author{Karl Bringmann \thanks{Saarland University and Max Planck Institute for Informatics, Saarland Informatics Campus, Germany. \texttt{bringmann@cs.uni-saarland.de}}
\quad
Vasileios Nakos \thanks{Saarland University and Max Planck Institute for Informatics, Saarland Informatics Campus, Germany, \texttt{vnakos@mpi-inf.mpg.de}}}


\date{}
\maketitle


\medskip

\begin{abstract}
  Approximating \SubsetSum is a classic and fundamental problem in computer science and mathematical optimization.
  The state-of-the-art approximation scheme for \SubsetSum computes a $(1-\eps)$-approximation in time $\tOh(\min\{n/\eps, n+1/\eps^2\})$ [Gens, Levner'78, Kellerer et al.'97].
  In particular, a $(1-1/n)$-approximation can be computed in time $\Oh(n^2)$.
  
  We establish a connection to Min-Plus-Convolution, a problem that is of particular interest in fine-grained complexity theory and can be solved naively in time $\Oh(n^2)$.
  Our main result is that computing a $(1-1/n)$-approximation for \SubsetSum is subquadratically equivalent to Min-Plus-Convolution. 
  Thus, assuming the Min-Plus-Convolution conjecture from fine-grained complexity theory, there is no approximation scheme for \SubsetSum with strongly subquadratic dependence on $n$ and $1/\eps$.
  In the other direction, our reduction allows us to transfer known lower order improvements from Min-Plus-Convolution to \SubsetSum, which yields a mildly subquadratic randomized approximation scheme.
  This adds the first approximation problem to the list of problems that are equivalent to Min-Plus-Convolution.
  
  \smallskip
  For the related \Partition problem, an important special case of \SubsetSum, the state of the art is a randomized approximation scheme running in time $\tOh(n+1/\eps^{5/3})$ [Mucha~et~al.'19]. 
  We adapt our reduction from \SubsetSum to Min-Plus-Convolution to obtain a related reduction from \Partition to Min-Plus-Convolution. 
  This yields an improved approximation scheme for \Partition running in time $\tOh(n + 1/\eps^{3/2})$. Our algorithm is the first \emph{deterministic} approximation scheme for \Partition that breaks the quadratic barrier. 
\end{abstract}


\newpage

\section{Introduction}

\subsection{Approximating SubsetSum}

In the \SubsetSum problem the task is to decide, given a set $X$ of $n$ positive integers and a target number $t$, whether some subset of $X$ sums to $t$. This is a fundamental problem at the intersection of computer science, mathematical optimization, and operations research. It belongs to Karp's initial list of 21 NP-complete problems~\cite{Karp72}, is a cornerstone of algorithm design~\cite{KPP04book,MT90book}, and a special case of many other problems like \Knapsack or \textsc{Integer Programming}. 


Since the problem is NP-hard, it is natural and imperative to study approximation algorithms.
To this end, consider the following classic optimization version of \SubsetSum: Given a set $X$ of $n$ positive integers and a target $t$, compute the maximum sum among all subsets summing to at most~$t$. Formally, the task is to compute $\OPT := \max\{ \Sigma(Y) \mid Y \subseteq X,\, \Sigma(Y) \le t \}$, where $\Sigma(Y)$ denotes the sum of all elements of $Y$.
This optimization version is still a special case of \Knapsack and \textsc{Integer Programming}, and it naturally gives rise to the classic notion of an \emph{approximation scheme} for \SubsetSum: Given an instance $(X,t)$ and a parameter $\eps > 0$, compute a subset $Y \subseteq X$ that satisfies $(1-\eps)\OPT \le \Sigma(Y) \le \OPT$. 

In 1975 Ibarra and Kim~\cite{IbarraK75} designed the first approximation scheme for \SubsetSum, running in time $\Oh(n/\eps^2)$.
According to~\cite[Section 4.6]{KPP04book}, a similar algorithm was found by Karp~\cite{Karp75}.
Lawler then presented an $\Oh(n + 1/\eps^4)$-time algorithm~\cite{Lawler79}. These three algorithms in fact even work for the more general \Knapsack problem. For \SubsetSum, Gens and Levner designed further improved algorithms running in time $\Oh(n/\eps)$~\cite{gens1978approximation,GensL79} and in time $\Oh(\min\{n/\eps, n + 1/\eps^3\})$~\cite{gens1994fast}. Finally, Kellerer et al.~presented an approximation scheme running in time $\Oh(\min\{n/\eps, n + 1/\eps^2 \log (1/\eps)\})$~\cite{KellererPS97,KellererMPS03}.
This remains the state of the art for over 20 years.
In particular, expressing the running time in the form $\Oh((n+1/\eps)^c)$, the exponent $c=2$ was achieved over 40 years ago~\cite{gens1978approximation,GensL79}, but an exponent $c < 2$ remains elusive. The starting point of this paper is, thus, the following question:

\begin{center}
  \emph{ Does \SubsetSum admit an approximation scheme in time $\Oh((n+1/\eps)^{2-\delta})$ for some $\delta > 0$? }
\end{center}

\paragraph{Pseudopolynomial-time Algorithms.}
Observe that an approximation scheme with a setting of $\eps < 1/t$ solves \SubsetSum exactly. For this reason, approximation algorithms are closely related to exact algorithms running in pseudopolynomial time. 
The classic pseudopolynomial time algorithm for \SubsetSum is Bellman's dynamic programming algorithm that runs in time $\Oh(n t)$~\cite{bellman1957dynamic}. 
This running time was recently improved, first by Koiliaris and Xu to\footnote{By $\tOh$ we hide factors of the form $\polylog(n M)$, where $n$ is the number of input numbers and $M$ is the largest input number.} $\tOh( \mathrm{min}\{ \sqrt{n} \cdot t, t^{4/3} \})$~\cite{KoiliarisX17}, 
then further to a randomized $\tOh(n+t)$ algorithm~\cite{Bring17}, see also~\cite{JW19} for further improvements in terms of logfactors. 
The improved running time is optimal up to lower order factors; in particular there is no algorithm running in time $\poly(n)\, t^{0.999}$ assuming the Strong Exponential Time Hypothesis~\cite{ABHS19} or the Set Cover Conjecture~\cite{CyganDLMNOPSW16}. 

These developments for pseudopolynomial algorithm raise the hope that one can design better approximation schemes for \SubsetSum, improving time $\Oh(n/\eps)$ to $\tOh(n+1/\eps)$, analogously to the improvement from $\Oh(n t)$ to $\tOh(n + t)$. Such an improved running time of $\tOh(n+1/\eps)$ would then be optimal up to lower order factors, since the lower bound for pseudopolynomial algorithms transfers to approximation schemes, in particular \SubsetSum has no approximation scheme in time $\poly(n) / \eps^{0.999}$, assuming the Strong Exponential Time Hypothesis or the Set Cover Conjecture. 

It seems difficult to prove a higher lower bound, thereby ruling out an improvement to time $\tOh(n+1/\eps)$. Indeed, so far all conditional lower bounds for approximations schemes for \SubsetSum, \Knapsack, and related problems used the connection to pseudopolynomial time algorithms~\cite{MuchaW019}. Since \SubsetSum is in pseudopolynomial time $\tOh(n+t)$, this connection cannot prove any higher lower bounds than $n+1/\eps$, up to lower order factors. 

\paragraph{Min-Plus-Convolution.}
In this work we connect \SubsetSum to the \minconv problem, in which we are given integer sequences $A, B \in \Z^n$ and the goal is to compute the sequence $C \in \Z^{2n}$ with $C[k] = \min_{i+j=k} A[i] + B[j]$. The naive running time of $\Oh(n^2)$ can be improved to $n^2 / 2^{\Omega(\sqrt{\log n})}$ by a reduction to All Pairs Shortest Path~\cite{BremnerCDEHILPT14} and using Williams' algorithm for the latter~\cite{Williams18}. 
Despite considerable attention~\cite{BussieckHWZ94,BremnerCDEHILPT14,LaberRC14,ChanL15,KunnemannPS17,
BackursIS17,BateniHSS18,CyganMWW19,AxiotisT19,
BringmannKW19,JansenR19}, no $\Oh(n^{2-\delta})$-time algorithm has been found, which was formalized as a hardness conjecture in fine-grained complexity theory~\cite{KunnemannPS17,CyganMWW19}. Many conditional lower bounds from the \minconv conjecture as well as several \minconv-equivalent problems are known, see, e.g.,~\cite{LaberRC14,KunnemannPS17,BackursIS17,CyganMWW19,JansenR19}. In particular, the \Knapsack problem with weight budget $W$ can be solved in time $\Oh((n+W)^{2-\delta})$ for some $\delta > 0$ if and only if \minconv can be solved in time $\Oh(n^{2-\delta'})$ for some $\delta'>0$~\cite{KunnemannPS17,CyganMWW19}.

\paragraph{Our Contribution to SubsetSum.}
We prove that computing a $(1-1/n)$-approximation for \SubsetSum is equivalent to the \minconv problem, thus adding the first approximation problem to the list of known \minconv-equivalent problems.
This negatively answers our main question for \SubsetSum approximation schemes running in strongly subquadratic time, conditional on the \minconv conjecture.
Moreover, our reductions allow us to transfer the known lower order improvements from \minconv to approximating \SubsetSum, which yields the first algorithmic improvement in over 20 years. 

\subsection{Approximating Partition} 

The \Partition problem is the special case of \SubsetSum where we want to partition a given set~$X$ into two sets of the same sum, or, equivalently, solve \SubsetSum with target $t =\Sigma(X)/2$. 
This problem arises in scheduling~\cite{karmarkar1982probabilistic}, minimization of circuit sizes and cryptography~\cite{merkle1978hiding}, as well as game theory~\cite{hayes2002computing}.
While naturally all algorithms for \SubsetSum also apply to the \Partition problem, designing approximation algorithms specific to \Partition also has a long history, see, e.g.,~\cite{gens1980fast}. 
From a lower bounds perspective, the state of the art for \Partition is the same as for \SubsetSum: by the same arguments as above we cannot expect an approximation scheme running in time $\poly(n) / \eps^{0.999}$. 
However, from an upper bounds perspective the two problems differ significantly: For a long time the fastest approximation scheme for \Partition was also quadratic, namely in time $\tOh(\min\{n/\eps, n + 1/\eps^2\})$~\cite{gens1980fast}, but recently Mucha et al.~\cite{MuchaW019} designed a randomized approximation scheme with a \emph{strongly subquadratic} running time of $\tOh(n+ 1/\eps^{5/3})$. 

\paragraph{Our Contribution to Partition.}
Note that we separate \SubsetSum and \Partition with respect to approximation, since for the former we prove a quadratic lower bound based on the \minconv conjecture, while the latter has a strongly subquadratic approximation scheme~\cite{MuchaW019}.

Furthermore, we design an improved approximation scheme for \Partition that runs in time $\tOh(n + 1/\eps^{3/2})$, by extending our techniques that we develop for reducing \SubsetSum to \minconv. 
It is worth noting that our algorithm is \emph{deterministic}, and thus it is the first deterministic approximation scheme for \Partition with subquadratic running time, answering a question in~\cite{MuchaW019}.

\medskip
The unifying theme between our algorithmic results is a notion of \emph{sparse} approximation, along with an efficient algorithm which takes as input sparse approximations of two sets, and computes a sparse approximation of their sumset by calling as black box a \minconv algorithm.

\subsection{Formal Statement of Results}

\paragraph{Approximating SubsetSum.}

We prove a subquadratic equivalence of approximating \SubsetSum and \minconv, by designing two reductions as follows.

\begin{thm}[Reduction from \SubsetSum to \minconv, Section~\ref{sec:algorithm}] \label{thm:mainalgo}
  If \minconv can be solved in time $T(n)$, then \SubsetSum has a randomized approximation scheme that is correct with high probability\footnote{In this paper, ``with high probability'' means success probability $1 - (\eps/n)^c$ for a constant $c>0$ that can be freely chosen by adapting the constants in the algorithm.} and runs in time $\tOh(n + T(1/\eps))$.
\end{thm}

Our reduction even transfers the lower order improvements of the \minconv algorithm that runs in time $n^2 / 2^{\Omega(\sqrt{\log n})}$~\cite{BremnerCDEHILPT14,Williams18}. This yields the first improved approximation scheme for \SubsetSum in over 20 years.

\begin{cor}[Approximation Scheme for SubsetSum] \label{cor:algo}
  \SubsetSum has a randomized approximation scheme that is correct with high probability and runs in time 
  \[ \tOh\bigg(n + \frac{(1/\eps)^2}{2^{\Omega(\sqrt{\log(1/\eps)})}}\bigg). \]
\end{cor}

The second reduction is as follows:

\begin{thm}[Reduction from \minconv to \SubsetSum, Section~\ref{sec:lowerbound}] \label{thm:mainlowerbound}
  If \SubsetSum has an approximation scheme running in time $\Oh((n + 1/\eps)^{2-\delta})$ for some $\delta > 0$, then \minconv can be solved in time $\Oh(n^{2-\delta'})$ for some $\delta' > 0$.
\end{thm}

Under the \minconv conjecture this rules out the existence of approximation schemes for \SubsetSum running in strongly subquadratic time $\Oh((n+1/\eps)^{2-\delta})$ for any $\delta > 0$. Taken together, the two reductions prove that \SubsetSum has an approximation scheme in time $\Oh((n+1/\eps)^{2-\delta})$ for some $\delta>0$ if and only if \minconv can be solved in time $\Oh(n^{2-\delta'})$ for some $\delta'>0$.

%


\paragraph{Approximating Partition.}

We adapt our reduction from \SubsetSum to \minconv for \Partition and obtain the following reduction.

\begin{thm}[Reduction from \Partition to \minconv, Section~\ref{sec:apxpartition}]\label{thm:partition}
If \minconv can be solved in time $T(n)$, then \Partition has a deterministic approximation scheme that runs in time 
\[ \tOh\bigg(n + \min_{1 \le L \le 1/\eps} \bigg\{ L \cdot T\Big( \frac 1{L \eps} \Big) + \frac L \eps \bigg\} \bigg). \]
\end{thm}

Even with the trivial $\Oh(n^2)$-time algorithm for \minconv, for $L := 1/\eps^{1/2}$ this yields an improved approximation scheme for \Partition running in time $\tOh(n + 1/\eps^{3/2})$. Plugging in the best known bound of $n^2 / 2^{\Omega(\sqrt{\log n})}$ for \minconv~\cite{BremnerCDEHILPT14,Williams18}, which has been derandomized in~\cite{ChanW16}, yields the following.

\begin{cor}[Approximation Scheme for Partition]\label{cor:partition}
\Partition has a deterministic approximation scheme that runs in time \[\tOh\left(n + \frac{(1/\eps)^{3/2}}{2^{\Omega(\sqrt{\log(1/\eps)}) }} \right).\]
\end{cor}

In comparison with the previously best approximation scheme by Mucha et al.~\cite{MuchaW019}, we improve the exponent of $1/\eps$ from $5/3$ to $3/2$, and our algorithm is \emph{deterministic}. Moreover, while their algorithm makes black-box usage of a relatively involved additive-combinatorics-based algorithm by Galil and Margalit~\cite{galil1991almost}, we obtain our results via elementary means, closer in spirit to the traditional approaches for approximating \SubsetSum.
In fact, we obtain our approximation scheme for \Partition by augmenting our reduction from \SubsetSum to \minconv by addtional ideas specific to the \Partition problem.

\subsection{Technical Overview}
\paragraph{Reduction from \minconv to \SubsetSum.} In our lower bound for \SubsetSum (Theorem~\ref{thm:mainlowerbound}), we start from the known reduction from \minconv to \Knapsack~\cite{KunnemannPS17,CyganMWW19}, and design a surprisingly simple reduction from (exact) \Knapsack to approximating \SubsetSum. In this reduction, we use the strict condition $\Sigma(Y) \le t$ in an interesting way, to simulate exact inequality checks on sums of very large numbers, despite being in an approximate setting. This allows us to embed the potentially very large values of \Knapsack items into an instance of approximate \SubsetSum.

\medskip
\paragraph{Reduction from \SubsetSum to \minconv.} The other reduction is essentially an approximation scheme for \SubsetSum, using as black box an algorithm for \minconv with subquadratic running time $\Tminconv(n)$.

It seems difficult to adapt the known approximation schemes~\cite{gens1978approximation,GensL79,gens1994fast,KellererPS97,KellererMPS03} 
to make use of a \minconv algorithm, since they all in some way follow the same approach as Bellman's pseudopolynomial algorithm: Writing $X = \{x_1,\ldots,x_n\}$, they compute an approximation of the set of all subset sums of $\{x_1,\ldots,x_i\}$ from an approximation of the set of all subset sums of $\{x_1,\ldots,x_{i-1}\}$, for $i=1,\ldots,n$. 
To obtain total time $\tOh(n/\eps)$, each iteration must run in time $\tOh(1/\eps)$, so there is no point at which a subquadratic algorithm for \minconv seems useful. (The $\tOh(n + 1/\eps^2)$-time approximation schemes follow the same iteration, but start with a preprocessing step that removes all but $\tOh(1/\eps)$ items.)

In contrast, the recent pseudopolynomial algorithms for \SubsetSum~\cite{KoiliarisX17,Bring17,JW19} use convolution methods, so in principle their structure allows to plug in a \minconv algorithm. Moreover, the running time $\tOh(n+t)$~\cite{Bring17,JW19} suggests to replace standard convolution in time $\tOh(t)$ by \minconv in time $\Tminconv(1/\eps)$, to obtain the desired running time of $\tOh(n + \Tminconv(1/\eps))$. However, all previous algorithms along this line of research heavily assume an exact setting, specifically that we have computed exact solutions to subproblems.

Here, we bring these two approaches together, by \emph{using ideas from the known approximation schemes to define the right notion of approximation, and then following the high-level structure of the pseudopolynomial algorithms with adaptations for our notion of approximation.}

In our notion of approximation, we say that a set $A$ approximates a set $B$ if $A \subseteq B$ and for any $b \in B$ there are lower and upper approximations $a^-,a^+ \in A$ with $a^- \le b \le a^+ \le a^- + \eps t$. To avoid having to solve \SubsetSum exactly, we need to relax this notion further, by allowing $a^+$ to take the value $t+1$, for details see Definition~\ref{def:apx}.
We establish that this notion satisfies several natural properties, e.g., it is transitive and behaves nicely under unions and sumsets. 

The connection to \minconv is then as follows. The main subroutine used by the recent pseudopolynomial algorithms is \emph{sumset computation}: Given $A,B \subseteq \N$, compute $A+B = \{a+b \mid a \in A,\, b \in B\}$. Since the output-size of sumset computation can be quadratic in the input-size, here we relax the goal and design a subroutine for \emph{approximate sumset computation}, computing a set $R$ that approximates $A+B$. 
To implement this subroutine, we first define a rasterization of $A$ and $B$ as vectors $A',B'$ with 
\[ A'[i] := \min\big(A \cap [i \eps t/2, (i+1)\eps t/2]\big) \qquad B'[j] := \min\big(B \cap [j \eps t/2, (j+1)\eps t/2]\big). \]
We then compute the vector $C'$ as the \minconv of $A'$ and $B'$, that is,
\[ C'[k] = \min_{i+j=k} A'[i] + B'[j]. \]
Note that we used the operation $\min$ at three positions in the above equations. By replacing some of them by $\max$, we obtain $2^3 = 8$ similar expressions, giving rise to vectors $C'_1,\ldots,C'_8$. 
We show that \emph{the set of all entries} of $C'_1,\ldots,C'_8$ approximates $A+B$ according to our notion of approximation. 
Since all involved vectors have length $\Oh(1/\eps)$, we can approximate sumsets in time $\Oh(\Tminconv(1/\eps))$.

Finally, we use this approximate sumset computation as a subroutine and follow (a simplified variant of) the pseudopolynomial algorithm from~\cite{Bring17}. The pseudocode is not changed much compared to~\cite{Bring17}, but the correctness proofs are significantly more complicated due to our approximate setting.

\paragraph{Approximating Partition.} 
Here we describe our improved approximation scheme for \Partition, which follows from augmenting our reduction from \SubsetSum to \minconv.



Recall that \Partition is the same as \SubsetSum on target $t = \Sigma(X)/2$, and for simplicity let us assume $\OPT = t$.
Previous approximation schemes for this problem made the following observations. First, although the goal is to find a subset $Y\subseteq X$ with $(1-\eps)t \le \Sigma(Y) \le t$, it suffices to relax the goal to a weak approximation with $(1-\eps)t \le \Sigma(Y) \le (1+\eps)t$. Indeed, if $t < \Sigma(Y) \le (1+\eps)t$, then its complement $X \setminus Y$ satisfies the desired guarantee $(1-\eps)t \le \Sigma(X \setminus Y) \le t$; this heavily uses that for \Partition we have $t = \Sigma(X)/2$.
Second, because we can relax to this weak, two-sided approximation guarantee, we may round the input numbers. Specifically, if we focus on subsets $Y \subseteq X$ of small cardinality $|Y| \le L$, then rounding every number $x \in X$ to $\left\lfloor x \cdot \frac L{\eps t}\right\rfloor$ and solving \SubsetSum in the rounded instance suffices for obtaining the weak approximation guarantee. 

The approximation scheme by Mucha et al.~\cite{MuchaW019} uses these observations by splitting $X$ into large and small numbers. Any solution can only contain few large numbers, so in this part they can focus on subsets of small cardinality, which allows them to use the above rounding procedure and to solve the rounded instance by the pseudopolynomial algorithm~\cite{Bring17}. 
For the small numbers, they exploit the solution structure of the instance by invoking a relatively involved additive-combinatorics-based algorithm by Galil and Margalit~\cite{galil1991almost}.


We avoid this dichotomy, and instead use the following observation: If we have sets $Z_1,Z_2,\ldots,Z_L$ such that $(Z_1+Z_2+\ldots+Z_L)\cap [t]$ approximates $\mathcal{S}(X;t)$ well, then we can again perform rounding as before. Specifically, by rounding every $x \in \bigcup_{i=1}^L Z_i$ to $\left\lfloor x \cdot \frac L{\eps t}\right\rfloor$ and solving the problem on the rounded instance, we obtain a value (and the corresponding set) that additively approximates $\OPT$ up to $\pm \eps t$ in a weak sense. The running time is controlled by the number $L$ of sets $Z_i$, as well as the time to solve the rounded instance. Note, however, that the rounded instance is \emph{not} a \SubsetSum instance as it was in~\cite{MuchaW019} (because $Z_i$ may not be a two-element set), and hence it cannot be solved by a black-box invocation of~\cite{Bring17}. However, it can still be solved using the standard exact sumset computation via Fast Fourier Transform in a tree-like fashion; the running time then depends on the structure of the sets $Z_i$; in particular on the sizes of the intermediate sumsets.

In our approximation scheme for \Partition we use this observation to augment our algorithm for \SubsetSum. 
Specifically, in our approximation scheme for \SubsetSum we build a recursion tree where level $\ell$ corresponds to $2^\ell$ many approximate sumset computations of sets with a sparse approximation of size $1/(2^\ell \eps)$. Since our algorithm for approximate sumset computation requires essentially quadratic time, each recursive call on level $\ell$ takes time $\tOh(1/(2^\ell \eps)^2)$, which over all $2^\ell$ recursive calls on level $\ell$ gives time $\tOh(1/(2^\ell \eps^2))$. 
On the lower levels of this recursion tree, where $\ell$ is large, 
this running time is low. 
On the upper levels, where $\ell$ is small, we are left with the task of combining $L = 2^\ell$ subproblems. This is exactly the setting $(Z_1+\ldots+Z_L) \cap [t]$ as discussed above. 
Hence, we replace the upper levels of the recursion tree by rounding and exact sumset computation. By balancing the lower and upper part, we obtain an approximation scheme for \Partition running in time $\tOh(n + 1/\eps^{3/2})$.

\subsection{Related Work}

Our reduction from \minconv to approximating \SubsetSum follows a recent trend in fine-grained complexity theory to show hardness of approximation. 
The first such result was presented by Abboud et al.~\cite{AbboudRW17}, who proved a PCP-like theorem for problems in P and obtained hardness of approximation for Orthogonal Vectors. 
Their results were extended to an equivalence and improved quantitatively (see, e.g.,~\cite{Chen18,ChenW19,ChenGLRR19}) and generalized to parameterized complexity (see, e.g.,~\cite{ChalermsookCKLM17,SLM19}). A similar approach was used on All Pairs Shortest Path~\cite{BringmannKW19}. 
While this line of research developed techniques to prove conditional lower bounds for constant-factor approximation algorithms (and higher approximation ratios), 
in this paper we obtain a conditional lower bound for an approximation scheme, which does not already follow from a lower bound for a constant-factor approximation. 

\smallskip
Approximations schemes for the related \Knapsack problem have also been widely studied, see, e.g.,~\cite{IbarraK75,Lawler79,KPP04book,gens1980fast,Chan18a,Jin19}.
For \Knapsack, the state-of-the-art approximation scheme runs in time $\tOh(n + 1/\eps^{9/4})$~\cite{Jin19}, see also~\cite{Chan18a}. A lower bound of $(n+1/\eps)^{2-o(1)}$ holds assuming the \minconv conjecture~\cite{MuchaW019}; our conditional lower bound in this paper can be seen as an extension of this result to \SubsetSum.
Note that our \SubsetSum result in this paper yields the first matching upper and lower bounds for approximation schemes for one of the classic \Knapsack-type problems \textsc{Partition}, \SubsetSum, and \Knapsack.

\smallskip
In this paper we consider variants of approximate \SubsetSum that keep the strict constraint $\Sigma(Y) \le t$ intact. 
Mucha et al.~\cite{MuchaW019} introduced a weaker variant of approximating \SubsetSum, where they also relax this constraint from $\Sigma(Y) \le t$ to $\Sigma(Y) \le (1+\eps) t$. They showed that such a weak approximation is closely connected to approximating \Partition, and they obtain the same running time $\tOh(n+1/\eps^{5/3})$ for both tasks. This shows that the strict upper bound $\Sigma(Y) \le t$ is crucial for our results. Indeed, their algorithm and our conditional lower bound separate their weak notion of approximation from the classic variant of approximation studied in this paper.
It seems that it is possible to extend our approximation scheme for \Partition to a weak approximation for \SubsetSum, at the cost of making it randomized. We leave it as an open problem to find a deterministic weak approximation scheme for \SubsetSum running in time $\tOh(n + 1/\eps^{3/2})$.


\smallskip
Further related work on the \SubsetSum problem designs pseudopolynomial time algorithms for \textsc{ModularSubsetSum}~\cite{AxiotisBJTW19}, tackles the problem whether all attainable subset sums smaller than~$t$ can be efficiently listed in outputsensitive time~\cite{BN20}, and considers many more directions, see, e.g.,~\cite{AustrinKKN15,AustrinKKN16,LincolnWWW16,BansalGN018,Nederlof17,EsserM19arxiv,BN20}.

\section{Preliminaries}

We write $\N = \{0,1,2,\ldots\}$. For $t \in \N$ we let $[t] = \{0,1,\ldots,t\}$. For sets $A,B \subseteq \N$ we define their \emph{sumset} as $A+B = \{a+b \mid a \in A,\, b \in B\}$ and their \emph{capped sumset} as $A \plt B = (A+B) \cap [t]$. 

We use $\Sigma(Y)$ as shorthand notation for $\sum_{y \in Y} y$, and we denote the set of all subset sums of $X$ below $t$ by $\cS(X;t) := \left\{\Sigma(Y) \mid Y \subseteq X,\, \Sigma(Y) \le t\right\}$. Furthermore, we denote $\cS(X) := \cS(X;\infty) = \{ \Sigma(Y) \mid Y \subseteq X\}$. 

Recall that in \minconv we are given integer sequences $A = (A[0],\ldots,A[n-1])$ and $B = (B[0],\ldots,B[n-1])$ and the goal is to compute the sequence $C = (C[0],\ldots,C[2n-1])$ with $C[k] = \min_{i+j=k} A[i] + B[j]$, where the minimum ranges over all pairs $(i,j)$ with $0 \le i,j < n$ and $i+j=k$. We will also consider the equivalent \maxconv problem, where the task is to compute $C'[k] = \max_{i+j=k} A[i] + B[j]$.
We assume that all entries of $A$ and $B$ are from some range $\{1,\ldots,M\}$ and that arithmetic operations on $\Oh(\log M + \log n)$-bit numbers can be performed in constant time. This holds, e.g., in the RAM model with $\Theta(\log n)$-bit memory cells if $M = n^{\Oh(1)}$.

Throughout the paper, by $\tOh$-notation we hide factors of the form $\polylog(n,M)$, where $n$ is the number of input numbers and $M$ is the largest input number.
For technical reasons, we assume all time bounds $T(\cdot)$ to satisfy $T(\Oh(n)) = \Oh(T(n))$ (and similar for multivariate functions). This is a natural assumption in the polynomial time world as well as in the pseudopolynomial setting.

\section{Lower Bound for SubsetSum}
\label{sec:lowerbound}


We now present our reduction from \minconv to approximating \SubsetSum, proving Theorem~\ref{thm:mainlowerbound}. In this section we consider the following formulation of approximate \SubsetSum:

\begin{center}
  \DisSubsetSum: Given $X,t$ and $\eps>0$, distinguish whether $\OPT = t$ or $\OPT < (1-\eps)t$. \\ If $\OPT \in [(1-\eps)t,t)$ the output can be arbitrary.
\end{center}

Note that any approximation scheme for \SubsetSum would be able to solve the above gap problem. By using this formulation as a gap problem, our reduction is general enough to cover a wide variety of specific problem formulations, as discussed in Section~\ref{sec:discussionproblems}.

\smallskip
Our reduction from \minconv to \DisSubsetSum goes via the \Knapsack problem: Given $n$ items with weights $w_1,\ldots,w_n$ and values $v_1,\ldots,v_n$, and given a weight budget $W$ and a value goal~$V$, decide whether for some $S \subseteq [n]$ we have $\sum_{i \in S} w_i \le W$ and $\sum_{i \in S} v_i \ge V$. 
We again denote by $M$ the largest input number, that is $w_i,v_i,W,V \in \{1,\ldots,M\}$.

Bellman's classic dynamic programming algorithm solves Knapsack in time $\Oh(n W)$~\cite{bellman1957dynamic}. In the regime $W \approx n$, this running time is $\Oh(n^2)$ and any improvement to time $\tOh(n^{2-\eps})$ would violate the \minconv conjecture, as was shown independently by Cygan et al.~\cite[Theorem 2]{CyganMWW19} and K\"unnemann et al.~\cite[Theorem 4.8]{KunnemannPS17}. Specifically, both references contain the following reduction.

\begin{thm} \label{thm:refKnapsackFromConvolution}
  If \Knapsack can be solved in time $T(n,W)$, then \minconv can be solved in time $\tOh(T(\sqrt n, \sqrt n) \cdot n)$.
  %
\end{thm}

We next show a simple reduction from (exact) \Knapsack to \DisSubsetSum. 

\begin{thm}
  If \DisSubsetSum can be solved in time $T(n,1/\eps)$, then \Knapsack can be solved in time $\tOh(T(n, W) + W)$.
\end{thm}
In combination with Theorem~\ref{thm:refKnapsackFromConvolution}, we obtain that if \DisSubsetSum can be solved in time $\tOh((n+1/\eps)^{2-\delta})$ for some $\delta > 0$, then \minconv can be solved in time $\tOh(\sqrt{n}^{2-\delta} \cdot n) = \tOh(n^{2-\delta/2})$. This proves Theorem~\ref{thm:mainlowerbound}. 

\begin{proof}
  Given a \Knapsack instance, if $n \le \log M$ then we run Bellman's algorithm to solve the instance in time $\Oh(n W) = \Oh(W \log M) = \tOh(W)$. Therefore, from now on we assume $n \ge \log M$.
  
  We construct an intermediate \Knapsack instance by adding items of weight $2^i$ and value $0$ for each $0 \le i \le \log(W)$, and adding items of weight 0 and value $-2^i$ for each $0 \le i \le \log(V)$. Since items of value less than or equal to 0 can be ignored, these additional items do not change whether the \Knapsack instance has a solution. However, in case there is a solution, we can use these additional items to fill up the total weight to exactly $W$ and decrease the total value to exactly $V$. In other words, if the \Knapsack instance has a solution, then it also has a solution of total weight $W$ and total value $V$.
  This increases the number of items by an additional $\Oh(\log M) = \Oh(n)$, since as discussed above we can assume $n \ge \log M$. As this is only an increase by a constant factor, with slight abuse of notation we still use $n$ to denote the number of items, and we denote the weights and values of the resulting items by $w_1,\ldots,w_n$ and $v_1,\ldots,v_n$, respectively. 
  
   We note that negative values are only used in the intermediate \Knapsack instance, and that all weights and values are still bounded by $M$, now in absolute value.
   The constructed \SubsetSum instance does not contain any negative numbers, as can be checked from the following construction.
  
  \medskip
  We set $M' := 4 n M$ and let $X$ consist of the numbers $x_i := w_i \cdot M' - v_i$ for $1 \le i \le n$. For $t := W \cdot M' - V$ and $\eps := 1/(2W)$ we construct the \DisSubsetSum instance $(X,t,\eps)$.
  
  \medskip
  To argue correctness of this reduction, first suppose that the \Knapsack instance has a solution. Using the added items, there is a set $S \subseteq [n]$ such that $\sum_{i \in S} w_i = W$ and $\sum_{i \in S} v_i = V$. The corresponding subset $\{x_i \mid i \in S\} \subseteq X$ thus sums to $W \cdot M' - V = t$. Hence, if the \Knapsack instance has a solution, then the constructed \SubsetSum instance satisfies $\OPT = t$.
  
  For the other direction, suppose that the \Knapsack instance has no solution, and consider any set $S \subseteq [n]$. Then we are in one of the following three cases:
  \begin{itemize}
  \item \emph{Case 1:} $\sum_{i \in S} w_i \ge W+1$. Then we have $\sum_{i \in S} x_i \ge (W+1) M' - \sum_{i \in S} v_i$. Since $\sum_{i \in S} v_i \le n M = M'/4$ we obtain $\sum_{i \in S} x_i > W M' > t$. Hence, $S$ does not correspond to a feasible solution of the constructed \SubsetSum instance.
  
  \item \emph{Case 2:} $\sum_{i \in S} w_i = W$. Since the \Knapsack instance has no solution, we have $\sum_{i \in S} v_i < V$. This yields $\sum_{i \in S} x_i = W M' - \sum_{i \in S} v_i > W M' - V = t$. Hence, $S$ does not correspond to a feasible solution of the constructed \SubsetSum instance.
  
  \item \emph{Case 3:} $\sum_{i \in S} w_i \le W-1$. Then we have $\sum_{i \in S} x_i \le (W-1) M' - \sum_{i \in S} v_i$. Since $|\sum_{i \in S} v_i| \le n M = M'/4$, we obtain $\sum_{i \in S} x_i \le W M' - 0.75 M'$. Using $V \le M \le M'/4$ yields $\sum_{i \in S} x_i \le W M' - V - 0.5 M'$. Finally, since $0.5 M' = \eps W M' > \eps (W M' - V)$, we obtain  $\sum_{i \in S} x_i < (1-\eps)(W M' - V) = (1-\eps) t$.
%
  \end{itemize}
  From this case distinction we obtain that any subset $Y \subseteq X$ is either infeasible, i.e., $\Sigma(Y) > t$, or sums to less than $(1-\eps)t$. Hence, if the \Knapsack instance has no solution, then the constructed \SubsetSum instance satisfies $\OPT < (1-\eps)t$.
  
  It follows that solving \DisSubsetSum on the constructed instance decides whether the given \Knapsack instance has a solution, which shows correctness. 
  
  Finally, we analyze the running time. 
  Since $\eps = 1/(2W)$, invoking a $T(n,1/\eps)$-time \DisSubsetSum algorithm on the constructed instance takes time $T(n,2W) = \Oh(T(n,W))$. 
  Together with the first paragraph, we obtain total time $\Oh(T(n,W) + W \log M)$ for \Knapsack.
%
%
\end{proof}

\begin{rem} \normalfont
  \Knapsack can be solved in time $\tOh(n+W^2)$ as follows. Note that any solution contains at most $W/w$ items of weight exactly $w$. We may therefore remove all but the $W/w$ most profitable items of any weight $w$. The number of remaining items is at most $\sum_{w=1}^W W/w = \tOh(W)$. After this $\tOh(n)$-time preprocessing, the classic $\Oh(nW)$-time algorithm runs in time $\tOh(W^2)$.
  
  In particular, \Knapsack is in time $\tOh(\min\{nW, n+W^2\})$. This nicely corresponds to the best-known running time of $\tOh(\min\{n/\eps, n+1/\eps^2\})$ for \ApxSubsetSum~\cite{gens1978approximation,GensL79,KellererPS97,KellererMPS03}. Our reduction indeed transforms the latter time bounds into the former. 
\end{rem}

\section{Algorithm for Approximating SubsetSum}
\label{sec:algorithm}

We now show our reduction from approximating \SubsetSum to exactly solving \minconv, proving Theorem~\ref{thm:mainalgo}. This implies the approximation scheme of Corollary~\ref{cor:algo}.
In this section, we work with the following formulation of approximate \SubsetSum:

\begin{center}
  Given $X,t$ and $\eps>0$, return any subset $Y \subseteq X$ satisfying \\ $\Sigma(Y) \le t$ and $\Sigma(Y) \ge \min\{\OPT, (1-\eps)t\}$.
\end{center}

Note that this is quite a strong problem variant, since for $\OPT \le (1-\eps)t$ it requires us to solve the problem exactly! In particular, an algorithm for the above problem variant implies a standard approximation scheme for \SubsetSum, see also the discussion in Section~\ref{sec:discussionproblems}.

\smallskip
From a high-level perspective, we first use ideas from the known approximation schemes for \SubsetSum to carefully define the right notion of approximation and to establish its basic properties (Section~\ref{sec:preparations}). 
With respect to this notion of approximation, we then design an approximation algorithm for sumset computation, using an algorithm for \minconv as a black box (Section~\ref{sec:apxsumset}).
Finally, we adapt the exact pseudopolynomial algorithm for \SubsetSum from~\cite{Bring17} by replacing its standard sumset subroutine by our novel approximate sumset computation, in addition to several further changes to make it work in the approximate setting (Section~\ref{sec:algo}). This yields an approximation scheme for \SubsetSum, with black-box access to a \minconv algorithm.

\smallskip
Throughout this section we assume to have access to an algorithm for \minconv running in time $\Tminconv(n)$ on sequences of length $n$. 
Since this is an exact algorithm, we can assume that $\Tminconv(n) = \Omega(n)$.

\subsection{Preparations}
\label{sec:preparations}

We start by defining and discussing our notion of approximation. We will set $\Delta = \eps t$ in the end.

\begin{defn}[Approximation] \label{def:apx}
  Let $t,\Delta \in \N$. For any $A \subseteq [t]$ and $b \in \N$, we define the \emph{lower and upper approximation of $b$ in $A$ (with respect to universe $[t]$)} as
  \[ \Apxm(b,A) := \max\{a \in A \cup \{t+1\} \mid a \le b\} \quad \text{and}\quad \Apxp(b,A) := \min\{a \in A \cup \{t+1\} \mid a \ge b\}. \]
  We use the convention $\max \emptyset := -\infty$ and $\min \emptyset := \infty$.
  
  For $A,B \subseteq \N$, we say that \emph{$A$ \Dtapx $B$} if $A \subseteq B \subseteq [t]$ and for any $b \in B$ we have 
  \[ \Apxp(b,A) - \Apxm(b,A) \le \Delta. \]
\end{defn}

Note that the approximations of $b$ in $A$ are not necessarily elements of $A$, since we add $t+1$.
We will sometimes informally say that ``$b$ has good approximations in $A$'', with the meaning that $\Apxp(b,A) - \Apxm(b,A) \le \Delta$ holds, where $t,\Delta$ are clear from the context.

There are two subtle details of this definition. First, we require $\Apxp(b,A) - \Apxm(b,A) \le \Delta$ instead of the more usual $\Apxp(b,A) - b \le \Delta$ and $b - \Apxm(b,A) \le \Delta$. For an example where this detail is crucially used see the proof of Lemma~\ref{lem:transitivity} below. This aspect of our definition is inspired by the approximation algorithm of Kellerer et al.~\cite{KellererPS97,KellererMPS03}.

Second, we change the upper end by adding $t+1$ to $A$. This relaxation is necessary because our goal will be to compute a set $A$ that \Dtapx the set of all subset sums of $X$ below $t$, or more precisely the set $\cS(X;t) := \{\Sigma(Y) \mid Y \subseteq X, \, \Sigma(Y) \le t \}$. Computing $\max(\cS(X;t))$ means to solve \SubsetSum exactly and is thus NP-hard. Therefore, we need a notion of approximation that does not force us to determine $\max(\cS(X;t))$, which is achieved by relaxing the upper end.

We start by establishing some basic properties of our notion of approximation. 

\begin{lem}[Transitivity] \label{lem:transitivity}
  If $A$ \Dtapx $B$ and $B$ \Dtapx $C$, then $A$ \Dtapx~$C$.
\end{lem}
\begin{proof}
  Since $A\subseteq B$ and $B \subseteq C$ we obtain $A \subseteq C$. 
  For any $c \in C$, let $b^-$ and $b^+$ be the lower and upper approximations of $c$ in $B$. Note that $b^-,b^+ \in B \cup \{t+1\}$. For any $b \in B$, since $A$ \Dtapx $B$ we find good approximations of $b$ in $A$. Additionally, for $b = t+1$ one can check that $\Apxp(b,A) = \Apxm(b,A) = t+1$, and thus we also have $\Apxp(b,A) - \Apxm(b,A) \le \Delta$. Therefore, $b^-$ and $b^+$ both have good approximations in $A$.
  So let $a^{--}$ and $a^{-+}$ be the lower and upper approximations of $b^-$ in~$A$, and similarly let $a^{+-}$ and $a^{++}$ be the lower and upper approximations of $b^+$ in $A$. 
  If $c \le a^{-+}$, then $a^{--} \le c \le a^{-+}$ form approximations of $c$ in $A$ within distance $\Delta$. Similarly, if $a^{+-} \le c$, then $a^{+-} \le c \le a^{++}$ form approximations of $c$ in $A$ within distance $\Delta$. In the remaining case we have
  $$ b^- \le a^{-+} \le c \le a^{+-} \le b^+. $$
  It follows that $a^{-+} \le c \le a^{+-}$ form approximations of $c$ in $A$ that are within distance $\Delta$, since they are sandwiched between $b^-$ and $b^+$.
\end{proof}

\begin{lem} \label{lem:sandwich}
  If $A \subseteq B \subseteq C$ and $A$ \Dtapx $C$, then $B$ \Dtapx $C$.
\end{lem}
\begin{proof}
  We have $B \subseteq C$, and for any $c \in C$ its approximations in $B$ are at least as good as in $A$.
\end{proof}

%

Our notion of approximation also behaves nicely under unions and sumsets, as shown by the following two lemmas.

\begin{lem}[Union Property] \label{lem:union}
  If $A_1$ \Dtapx $B_1$ and $A_2$ \Dtapx $B_2$, then $A_1 \cup A_2$ \Dtapx $B_1 \cup B_2$. 
\end{lem}
\begin{proof}
  Let $r \in \{1,2\},\; b \in B_r$. The approximations of $b$ in $A_1 \cup A_2$ are at least as good as in $A_r$.
\end{proof}

\begin{lem}[Sumset Property] \label{lem:apxsumset}
  If $A_1$ \Dtapx $B_1$ and $A_2$ \Dtapx $B_2$, then $A_1 \plt A_2$ \Dtapx $B_1 \plt B_2$. 
\end{lem}
\begin{proof}
  Since $A_r \subseteq B_r$ we obtain $A_1 \plt A_2 \subseteq B_1 \plt B_2$. 
  So consider any $b_1 \in B_1, b_2 \in B_2$ and set $b := b_1+b_2$. Let $a_1^-,a_1^+$ be the lower and upper approximations of $b_1$ in $A_1$, and let $a_2^-,a_2^+$ be the lower and upper approximations of $b_2$ in $A_2$. Consider the intervals
  \[ L := [a_1^- + a_2^-, a_1^+ + a_2^-] \cap [t+1] \quad \text{and} \quad R := (a_1^+ + a_2^-, a_1^+ + a_2^+] \cap [t+1]. \]
  Note that the endpoints of $L$ and $R$ are contained in $(A_1 \plt A_2) \cup \{t+1\}$. Moreover, the interval $L$ has length at most $a_1^+ - a_1^- \le \Delta$, and similarly $R$ has length at most $a_2^+ - a_2^- \le \Delta$. Finally, since 
  \[ a_1^- + a_2^- \le b_1 + b_2 \le a_1^+ + a_2^+, \]
  we have $b \in L \cup R$, and thus either $b \in L$ or $b \in R$. The endpoints of the respective interval containing $b$ thus form lower and upper approximations of $b$ in $A_1 \plt A_2$ within distance $\Delta$.
\end{proof}

We next show that we can always assume approximation sets to have small size, or more precisely to be locally \emph{sparse}.

\begin{defn}[Sparsity]
  Let $A \subseteq \mathbb{N}$ and $\Delta \in \N$. We say that $A$ is \emph{\Dsparse} if $|A \cap [x,x+\Delta]| \le 2$ holds for any $x \in \N$. 
  If $A$ is \Dsparse and $A$ \Dtapx $B$, we say that $A$ \emph{\sDtapx}~$B$.
\end{defn}

\begin{lem}[Sparsification] \label{lem:sparsification}
  Given $t,\Delta \in \N$ and a set\footnote{Here and in the following, we assume that input sets such as $B$ are given as a sorted list of their elements.} $B \subseteq [t]$, in time $\Oh(|B|)$ we can compute a set $A$ such that $A$ \sDtapx $B$.
\end{lem}
\begin{proof}
  Recall that our notion of approximation requires $A$ to be a subset of $B$.
  We inductively argue as follows. Initially, for $A := B$ it holds that $A$ \Dtapx $B$.  
  If there exist $a_1,a_2,a_3 \in A$ with $a_1 < a_2 < a_3 \le a_1+\Delta$, remove $a_2$ from $A$. 
  We claim that the resulting set $A$ still \Dtapx $B$. Indeed, consider any $b \in B$. If $b \le a_1$ we have $\Apxm(b,A) \le \Apxp(b,A) \le a_1$, and thus $a_2$ is irrelevant. Similarly, $a_2$ is also irrelevant for any $b \ge a_3$. Finally, for any $a_1 < b < a_3$, after removing $a_2$ we have $\Apxm(b,A) \ge a_1$ and $\Apxp(b,A) \le a_3$, and $a_3 - a_1 \le \Delta$. Thus, after removing $a_2$ the set $A$ still \Dtapx $B$. 
  Repeating this rule until it is no longer applicable yields a subset $A \subseteq B$ that contains at most two numbers in any interval $[x, x + \Delta]$. 
  
  Finally, it is easy to compute~$A$ in time $\Oh(|B|)$ by one sweep from left to right, assuming that $B$ is given in sorted order. Pseudocode for this is given in Algorithm~\ref{alg:sparsification}.
\begin{algorithm}
\caption{$\mathtt{Sparsification}(B,t,\Delta)$: Given $t,\Delta > 0$ and a set $B\subseteq [t]$ in sorted order, compute a set $A$ that \sDtapx $B$. We denote the elements of $B$ by $B[1],\ldots,B[m]$.}\label{alg:sparsification}
\begin{algorithmic}[1]
\State Initialize $A := \emptyset$ and $n := 0$
\For {$i=1,\ldots,m$}
  \State $n := n + 1$
  \State $A[n] := B[i]$
  \If{ $n \ge 3$ and $A[n] - A[n-2] \le \Delta$ }
    \State $A[n-1] := A[n]$
    \State $n := n - 1$
  \EndIf
\EndFor
\State \Return $\{A[1],\ldots,A[n]\}$
\end{algorithmic}
\end{algorithm}
\end{proof}

Lastly, we study how one can shift the value of $t$.

\begin{lem}[Down Shifting]
 \label{lem:downshift}
  Let $t,t',\Delta \in \N$ with $t \ge t'$. If $A$ \Dtapx $B$, then $A \cap [t']$ \Dtprimeapx $B \cap [t']$.
\end{lem}
\begin{proof}
  We write $A' := A \cap [t']$ and $B' := B \cap [t']$. Clearly we have $A' \subseteq B'$.
  For any $b \in B'$, note that 
  \[ \Apxmprime(b,A') = \Apxm(b,A), \] 
  since only elements larger than $b$ are removed from $A$ and $t+1 \ge t'+1 > b$. Moreover, if $\Apxp(b,A) \le t'$ then $\Apxp(b,A) \in A'$, and thus $\Apxpprime(b,A') = \Apxp(b,A)$. Otherwise, we have $\Apxpprime(b,A') \le t'+1 \le \Apxp(b,A)$. In any case, we have $\Apxpprime(b,A') \le \Apxp(b,A)$, and therefore
  $$ \Apxpprime(b,A') - \Apxmprime(b,A') \le \Apxp(b,A) - \Apxm(b,A) \le \Delta. \qedhere $$
\end{proof}

\begin{lem}[Up Shifting] \label{lem:upshift}
  Let $t,t',\Delta \in \N$ with $t \le t'$. If $A$ \Dtapx $B$, then there exists $u \in (t - \Delta, t]$ such that $A$ \Dtprimeapx $B \cap [u]$.
\end{lem}
\begin{proof}
  Note that $A \subseteq B \subseteq [t]$. We consider two cases.
  
  \emph{Case 1:} If $\max(A) \le t-\Delta$ then we set $u := t$, so that $B \cap [u] = B$. 
  For any $b \in B$ by assumption we have $\Apxp(b,A) - \Apxm(b,A) \le \Delta$. Since $\max(A) \le t-\Delta$, we have $\Apxm(b,A) \le t-\Delta$ and thus $\Apxp(b,A) \le t$. It follows that $\Apxp(b,A) \in A$ (instead of being $t+1$). Therefore, the same elements $\Apxm(b,A), \Apxp(b,A) \in A$ also form good approximations of $b$ in $A$ with respect to universe~$[t']$. Thus, $A$ \Dtprimeapx $B \cap [u]$.
  
  \emph{Case 2:} If $\max(A) \in (t-\Delta,t]$ then we set $u := \max(A)$. For any $b \in B \cap [u]$ we have $b \le \max(A)$, and thus the upper approximation of $b$ in $A$ is simply the smallest element of $A$ that is at least~$b$. As this is independent of the universe $[t]$, we obtain $\Apxp(b,A) = \Apxpprime(b,A)$. We also clearly have $\Apxm(b,A) = \Apxmprime(b,A)$. Hence, 
  $\Apxpprime(b,A) - \Apxmprime(b,A) = \Apxp(b,A) - \Apxm(b,A) \le \Delta$,
  so $A$ \Dtprimeapx $B \cap [u]$.
\end{proof}

\subsection{Algorithm for Approximate Sumset Computation}
\label{sec:apxsumset}

We now present the main connection to \minconv: We show how to compute for given $A_1,A_2$ a set $A$ that \apx $A_1+A_2$, by performing two calls to \minconv. At first we set $t := \infty$, so that we do not have to worry about the upper end. This will be fixed in Lemma~\ref{lem:cappedsumset} below.

\begin{lem}[Unbounded Sumset Computation] \label{lem:unbsumsetcomp}
  Given $t,\Delta \in \N$ with $t \ge \Delta$ and \Dsparse sets $A_1,A_2 \subseteq [t]$, in time $\Oh(\Tminconv(t/\Delta))$ we can compute a set~$A$ that \Dinftyapx $A_1 + A_2$.
\end{lem}
\begin{proof}
  To simplify notation, for this proof we introduce the symbol $\perp$ indicating an undefined value. We let $\min \emptyset = \max \emptyset = \perp$. Furthermore, we let $x + \!\!\perp = \perp$ and $\min\{x,\perp\} = \max\{x,\perp\} = x$. This gives rise to natural generalizations of \minconv and \maxconv to sequences over $\mathbb{Z} \cup \{\perp\}$. We call an entry of such a sequence \emph{defined} if it is not $\perp$.
  Note that since $\perp$ acts as a neutral element for the $\min$ and $\max$ operations, we can think of $\perp$ being $\infty$ for \minconv, and $-\infty$ for \maxconv. 
  The fact that these neutral elements, $\infty$ and $-\infty$, are different is the reason why we introduce $\perp$. 
  
Observe that if \minconv on sequences over $\{-M,\ldots,M\}$ is in time $\Tminconv(n)$, then also \minconv on sequences over $\{-M/4,\ldots,M/4\} \cup \{\perp\}$ is in time $\Oh(\Tminconv(n))$. Indeed, replacing $\perp$ by~$M$, any output value in $[-M/2,M/2]$ is computed correctly, while any output value in $[3M/4,2M]$ corresponds to~$\perp$. 
  Also observe that \maxconv is equivalent to \minconv by negating all input and output values, and therefore \maxconv is also in time $\Oh(\Tminconv(n))$.
  
  \medskip
  Our algorithm is as follows.
  Set $n := 4 \lceil t / \Delta \rceil$. We consider intervals $I_i := [i \Delta/2, (i+1) \Delta/2]$ for $0 \le i < n$.
  Since $A_1,A_2$ are \Dsparse, they contain at most two elements in any interval $I_i$. We may therefore ``unfold'' the sets $A_1,A_2$ into vectors $X_1,X_2$ of length $2n$ as follows. For $r \in \{1,2\}$ and $0 \le i < n$ we set
  \begin{align*}
    X_r[2i] &:= \min(I_i \cap A_r), \\
    X_r[2i+1] &:= \max(I_i \cap A_r).
  \end{align*}
  We then compute the sequences
  \begin{align*}
    C^- &:= \minconv(X_1, X_2), &&\text{that is, }\; C^-[k] = \min_{i+j=k} X_1[i]+X_2[j], \\
    C^+ &:= \maxconv(X_1, X_2), &&\text{that is, }\; C^+[k] = \max_{i+j=k} X_1[i]+X_2[j],
  \end{align*}
  for $0 \le k < 4n$.
  Finally, we return the set $A$ containing all defined entries of $C^-$ and $C^+$. 
  
  \medskip
  Clearly, this algorithm runs in time $\Oh(\Tminconv(t/\Delta))$. 
  It remains to prove correctness. 
  Since every defined entry of $X_r$ corresponds to an element of $A_r$, it follows that every defined entry of $C^-$ and $C^+$ corresponds to a sum in $A_1+A_2$. Hence, we have $A \subseteq A_1+A_2$.
  
  It remains to prove that for any $a_1 \in A_1, a_2 \in A_2$ their sum $a_1+a_2$ has good approximations in $A$. 
  Let $0 \le i^*,j^* < 2n$ be such that $X_1[i^*] = a_1$ and $X_2[j^*] = a_2$ and let $k^* := i^*+j^*$. Then by definition of \minconv and \maxconv we have
  \begin{align} \label{eq:myeqone} 
    C^-[k^*] \le a_1+a_2 \le C^+[k^*]. 
  \end{align}
  It remains to prove that $C^+[k^*] - C^-[k^*] \le \Delta$. 
  From the construction of $X_r[2i]$ and $X_r[2i+1]$ it follows that any defined entry satisfies $X_r[i] \in 
  [(i-1)\Delta/4,(i+1)\Delta/4]$. In particular, the sum of two defined entries satisfies $X_1[i] + X_2[j] \in [(i+j-2)\Delta/4,(i+j+2)\Delta/4]$. 
  This yields 
  $$ C^-[k^*], C^+[k^*] \in [(k^*-2)\Delta/4, (k^*+2)\Delta/4] \cup \{\perp\}. $$
  Moreover, at least one summand, $X_1[i^*] + X_2[j^*] = a_1+a_2$, is defined, and thus $C^-[k^*], C^+[k^*] \ne \perp$. 
  This yields $C^+[k^*] - C^-[k^*] \le \Delta$. Together with (\ref{eq:myeqone}) we see that any sum $a_1+a_2 \in A_1+A_2$ has good approximations in $A$, which finishes the proof.
\end{proof}

\begin{lem}[Capped Sumset Computation] \label{lem:cappedsumset}
  Let $t,\Delta \in \N$ and $B_1,B_2 \subseteq [t]$. Set $B := B_1 \plt B_2$ and suppose that $A_1$ \sDtapx $B_1$ and $A_2$ \sDtapx $B_2$. 
  In this situation, given $A_1,A_2,t,\Delta$, we can compute a set $A$ that \sDtapx $B$ in time $\Oh(\Tminconv(t/\Delta))$. We refer to this algorithm as $\mathtt{CappedSumset}(A_1,A_2,t,\Delta)$.
\end{lem}
\begin{proof}
  By Lemma~\ref{lem:apxsumset}, $A_1 \plt A_2$ \Dtapx $B$.
  Using Lemma~\ref{lem:unbsumsetcomp}, we can compute a set $A'$ that \Dinftyapx $A_1+A_2$. 
  By Lemma~\ref{lem:downshift} (Down Shifting), $A'' := A' \cap [t]$ \Dtapx $(A_1+A_2) \cap [t] = A_1 \plt A_2$. 
  Using Lemma~\ref{lem:sparsification}, given $A''$ we can compute a set $A$ that \sDtapx $A''$. By Lemma~\ref{lem:transitivity} (Transitivity), these three steps imply that $A$ \Dtapx $B$. Since $A$ is \Dsparse, $A$ also \sDtapx $B$.
  
  Each of these steps runs in time $\Oh(t/\Delta)$ or in time $\Oh(\Tminconv(t/\Delta))$. Since we can assume $\Tminconv(n) = \Omega(n)$, we can bound the total running time by $\Oh(\Tminconv(t/\Delta))$.
\end{proof}

\subsection{Algorithms for Subset Sum}
\label{sec:algo}

With the above preparations we are now ready to present our approximation algorithm for \SubsetSum. It is an adaptation of a pseudopolynomial algorithm for \SubsetSum~\cite{Bring17}, mainly in that we use Lemma~\ref{lem:cappedsumset} instead of the usual sumset computation by Fast Fourier Transform, but significant changes are required to make this work. 

Given $(X,t,\Delta)$ where $X$ has size $n$, our goal is to compute a set $A$ that \sDtapx the set $\cS(X;t) = \{\Sigma(Y) \mid Y \subseteq X, \, \Sigma(Y) \le t \}$.

\begin{defn}
  We say that an event happens \emph{with high probability} if its probability is at least $1 - \min\{1/n,\Delta/t\}^c$ for some constant $c>0$ that we are free to choose as any large constant. We say that \emph{$A$ \wDtapx $B$} if we have
  \begin{itemize}
    \item $A \subseteq B$, and
    \item with high probability $A$ \Dtapx $B$.
  \end{itemize}
\end{defn}

It can be checked that the properties established in Section~\ref{sec:preparations} still hold for this new notion of ``\wDtapx''. 

\subsubsection{Color Coding}

In this section, we loosely follow \cite[Section 3.1]{Bring17}.
We present an algorithm $\mathtt{ColorCoding}$ (Algorithm~\ref{alg:colorcoding}) which solves \SubsetSum in case all items are \emph{large}, that is, $X \subseteq [t/k,t]$ for a parameter~$k$. 


\begin{lem}[Color Coding] \label{lem:colorcoding}
  Given $t,\Delta,k \in \N$ with $t \ge \Delta$ and a set $X \subseteq [t/k,t]$ of size $n$, we can compute a set $A$ that  \wsDtapx $\cS(X;t)$, in time 
  $$\Oh\big(\big(n + k^2 \cdot \Tminconv(t/\Delta)\big) \log (nt/\Delta)\big).$$
\end{lem}
\begin{proof}
  Denote by $X_1,\ldots,X_{k^2}$ a random partitioning of $X$, that is, for every $x \in X$ we choose a number~$j$ uniformly and independently at random and we put $x$ into $X_j$. 
  For any subset $Y \subseteq X$ with $\Sigma(Y) \le t$, note that $|Y| \le k$ since $X \subseteq [t/k,t]$, and consider how the random partitioning acts on~$Y$.
  We say that the partitioning \emph{splits~$Y$} if we have $|Y \cap X_j| \le 1$ for any $1 \le j \le k^2$. By the birthday paradox, $Y$ is split with constant probability. More precisely, we can view the partitioning restricted to $Y$ as throwing $|Y| \le k$ balls into $k^2$ bins. Thus, the probability that $Y$ is split is equal to the probability that the second ball falls into a different bin than the first, the third ball falls into a different bin than the first two, and so on, which has probability
  $$ \frac{k^2-1}{k^2} \cdot \frac{k^2-2}{k^2} \cdots \frac{k^2-(|Y|-1)}{k^2} \ge \bigg( \frac{k^2-(|Y|-1)}{k^2} \bigg)^{|Y|} \ge \Big(1-\frac 1k\Big)^k \ge \Big( \frac 12 \Big)^2 = \frac 14. $$
  
  We make use of this splitting property as follows. Let $X'_j := X_j \cup \{0\}$ and 
  \[ T := X'_1 \plt \ldots \plt X'_{k^2}. \] 
  Observe that $T \subseteq \cS(X;t)$, since each sum appearing in $T$ uses any item $x \in X$ at most once. We claim that if $Y$ is split, then $T$ contains $\Sigma(Y)$. Indeed, in any part $X_j$ with $|Y \cap X_j| = 1$ we pick this element of $Y$, and in any other part we pick $0 \in X'_j$, to form $\Sigma(Y)$ as a sum appearing in $T = X'_1 \plt \ldots \plt X'_{k^2}$.
  
  Hence, we have $\Sigma(Y) \in T$ with probability at least $1/4$.
  To boost the success probability, we repeat the above random experiment several times. More precisely, for\footnote{Here $C$ is a large constant that governs the ``with high probability'' bound.} $r=1,\ldots,C \log(nt/\Delta)$ we sample a random partitioning $X = X_{r,1} \cup \ldots \cup X_{r,k^2}$, set $X'_{r,i} := X_{r,i} \cup \{0\}$, and consider $T_r := X'_{r,1} \plt \ldots \plt X'_{r,k^2}$. Since we have $\Sigma(Y) \in T_r$ with probability at least $1/4$, we obtain $\Sigma(Y) \in \bigcup_r T_r$ with high probability. Moreover, we have $\bigcup_r T_r \subseteq \cS(X;t)$.
  
  Let $\Ssp(X;t)$ be a sparsification of $\cS(X;t)$, and note that it has size $|\Ssp(X;t)| = \Oh(t/\Delta)$ and can be found in linear time by Lemma~\ref{lem:sparsification}. Since we use ``with high probability'' to denote a probability of at least $1 - \min\{1/n,\Delta/t\}^c$ for large constant $c$, we can afford a union bound over the $\Oh(t/\Delta)$ elements of $\Ssp(X;t)$ to infer that with high probability
  \[ \Ssp(X;t) \subseteq {\bigcup}_r T_r \subseteq \cS(X;t). \]
  Since $\Ssp(X;t)$ \Dtapx $\cS(X;t)$, Lemma~\ref{lem:sandwich} implies that
  \begin{align} \label{eq:Trapx}
    {\bigcup}_r T_r \;\;\text{\wDtapx}\;\; \cS(X;t). 
  \end{align}    
  
  We cannot afford to compute any $T_r$ explicitly, but we can compute approximations of these sets. To this end, let $Z_{r,j}$ be the sparsification of $X'_{r,j}$ given by Lemma~\ref{lem:sparsification}. We start with $A_{r,0} := \{0\}$ and repeatedly compute the capped sumset with $Z_{r,j}$, setting $A_{r,j} := \mathtt{CappedSumset}(A_{r,j-1},Z_{r,j},t,\Delta)$ for $1 \le j \le k^2$. 
  It now follows inductively from Lemma~\ref{lem:cappedsumset} that $A_{r,j}$ \sDtapx $X'_{r,1} \plt \ldots \plt X'_{r,j}$. 
  
  Hence, $A_{r,k^2}$ \sDtapx $T_r$. Let $A' := \bigcup_r A_{r,k^2}$. By Lemma~\ref{lem:union}, $A'$ \Dtapx $\bigcup_r T_r$. With (\ref{eq:Trapx}) and transitivity, $A'$ \wDtapx $\cS(X;t)$. Finally, we sparsify $A'$ using Lemma~\ref{lem:sparsification} to obtain a subset $A$ that \sDtapx $A'$. By transitivity, $A$ \wsDtapx $\cS(X;t)$. For pseudocode of this, see Algorithm~\ref{alg:colorcoding}.
  
  The running time is immediate from Lemmas~\ref{lem:sparsification} and~\ref{lem:cappedsumset}.
\end{proof}

\begin{algorithm}
\caption{$\mathtt{ColorCoding}(X,t,\Delta,k)$: Given $t,\Delta \in \N$ and a set $X\subseteq [t/k,t]$ in sorted order, we compute a set $A$ that \wsDtapx $\cS(X;t)$.}\label{alg:colorcoding}
\begin{algorithmic}[1]
\For {$r=1,\ldots,C \log (nt/\Delta)$}
  \State randomly partition $X = X_{r,1} \cup \ldots \cup X_{r,k^2}$
  \State $A_{r,0} := \{0\}$
  \For{$j=1,\ldots,k^2$}
    \State $X'_{r,j} := X_{r,j} \cup \{0\}$
    \State $Z_{r,j} := \mathtt{Sparsification}(X'_{r,j}, t, \Delta)$
    \State $A_{r,j} := \mathtt{CappedSumset}(A_{r,j-1}, Z_{r,j}, t, \Delta)$
  \EndFor
\EndFor
\State \Return $\mathtt{Sparsification}(\bigcup_r A_{r,k^2}, t, \Delta)$
\end{algorithmic}
\end{algorithm}

\subsubsection{Greedy}

We also need a special treatment of the case that all items are small, that is, $\max(X) \le \Delta$. In this case, we pick any ordering of $X = \{x_1,\ldots,x_n\}$ and let $P$ denote the set of all prefix sums $0, x_1, x_1+x_2, x_1+x_2+x_3, \ldots$ that are bounded by $t$, that is, $P = \{ \sum_{i=1}^j x_i \mid 0 \le j \le n\} \cap [t]$. We return a sparsification $A$ of $P$. See Algorithm~\ref{alg:greedy} for pseudocode.

\begin{claim} \label{cla:myniceclaim}
  $P$ \Dtapx $\cS(X;t)$.
\end{claim}
\begin{proof}
Clearly $P \subseteq \cS(X;t)$.
Moreover, any $s \in [0,\max(P)]$ falls into some interval between two consecutive prefix sums, and such an interval has length $x_i$ for some $i$. Hence, we have 
\[ \Apxp(s,P) - \Apxm(s,P) \le x_i \le \max(X) \le \Delta. \]

We now do a case distinction on $\Sigma(X)$.
If $\Sigma(X) < t$, then observe that $\max(P) = \Sigma(X) = \max(\cS(X;t))$. Therefore, the interval $[0,\max(P)]$ already covers all $s \in \cS(X;t)$ and we are done.

Otherwise, if $\Sigma(X) \ge t$, then observe that $\max(P) > t - \Delta$, as otherwise we could add the next prefix sum to $P$. In this case, for any $s \in [\max(P),t]$,
\[ \Apxp(s,P) - \Apxm(s,P) \le t+1 - \max(P) \le \Delta. \]
In total, every $s \in \cS(X;t)$ has good approximations in $P$.
\end{proof}

\begin{algorithm}
\caption{$\mathtt{Greedy}(X,t,\Delta)$: Given $t,\Delta \in \N$ and a set $X = \{x_1,\ldots,x_n\} \subseteq [t]$ with $\max(X) \le \Delta$, we compute a set $A$ that \sDtapx $\cS(X;t)$.}\label{alg:greedy}
\begin{algorithmic}[1]
\State $P := \{0\}$, $s := 0$, $i := 1$
\While{$i \le n$ and $s+x_i \le t$} 
  \State $s := s + x_i$
  \State $P := P \cup \{s\}$
  \State $i := i+1$
\EndWhile
\State $A := \mathtt{Sparsification}(P, t, \Delta)$
\State \Return $A$
\end{algorithmic}
\end{algorithm}

From Claim~\ref{cla:myniceclaim} and transitivity it follows for $A = \mathtt{Sparsification}(P, t, \Delta)$ that $A$ \sDtapx $\cS(X;t)$. We thus proved the following lemma.

\begin{lem}[Greedy] \label{lem:greedy}
  Given integers $t,\Delta > 0$ and a set $X \subseteq [t]$ of size $n$ satisfying $\max(X) \le \Delta$, we can compute a set $A$ that \sDtapx $\cS(X;t)$ in time $\Oh(n)$.
\end{lem}

\subsubsection{Recursive Splitting}

We now present a recursive algorithm making use of $\mathtt{ColorCoding}$ and $\mathtt{Greedy}$.
(This is a simplified version of the algorithms $\mathtt{FasterSubsetSum}$ and $\mathtt{ColorCodingLayer}$ from~\cite{Bring17}, adapted to the approximation setting.)

Given a set $X \subseteq \N$ of size $n$ and numbers $t,\Delta > 0$, our goal is to compute a set $A$ that \sDtapx $\cS(X;t)$.
We assume that initially $t \ge 8 \Delta$.
We will use parameters $k$ and $\eta$, which are set before the first call of the algorithm to\footnote{Here $C$ is a large constant that governs the ``with high probability'' bounds.} $k := \max\{8,C \log^3(nt/\Delta)\}$ and $\eta := 1/(2\log(t/\Delta))$.
We can assume that $X \subseteq [t]$, since larger numbers cannot be picked for subset sums in $[t]$. 

We partition $X$ into the large numbers $X_L := X \cap [t/k,t]$ and the small numbers $X_S := X \setminus X_L$. 
On the large numbers we compute $A_L := \mathtt{ColorCoding}(X_L,t,\Delta,k)$, so that $A_L$ \wsDtapx $\cS(X_L;t)$.
We then randomly partition the small numbers $X_S$ into subsets $X_1,X_2$, that is, for any $x \in X_S$ we choose a number $j \in \{1,2\}$ uniformly at random and we put $x$ into $X_j$. We recursively call the same algorithm on $(X_1,t',\Delta)$ and on $(X_2,t',\Delta)$ for the new target bound $t' := (1+\eta)t/2+\Delta$. Call the results of these recursive calls $A_1,A_2$. 
Finally, we combine $A_1,A_2$ to $A_S$, and $A_S,A_L$ to $A$, by capped sumset computations. We return $A$. The base case happens when $\max(X) \le \Delta$, where we run $\mathtt{Greedy}$.
See Algorithm~\ref{alg:finalalgo} for pseudocode.

\begin{algorithm}
\caption{$\mathtt{RecursiveSplitting}(X,t,\Delta)$: Given $t,\Delta \in \N$ and a set $X\subseteq [t]$ in sorted order, we compute a set $A$ that \sDtapx $\cS(X;t)$. The parameters $k,\eta$ are set before the first call of the algorithm to $k := \max\{8,C \log^3(nt/\Delta)\}$ and $\eta := 1/(2\log(t/\Delta))$.}\label{alg:finalalgo}
\begin{algorithmic}[1]
\If {$\max(X) \le \Delta$}
  \Return $\mathtt{Greedy}(X,t,\Delta)$
\EndIf
\State $X_L := X \cap [t/k,t]$, $X_S := X \setminus X_L$
\State randomly partition $X_S = X_1 \cup X_2$
\State $t' := (1+\eta)t/2+\Delta$
\State $A_L := \mathtt{ColorCoding}(X_L,t,\Delta,k)$
\State $A_1 := \mathtt{RecursiveSplitting}(X_1,t',\Delta)$
\State $A_2 := \mathtt{RecursiveSplitting}(X_2,t',\Delta)$
\State $A_S := \mathtt{CappedSumset}(A_1,A_2,t,\Delta)$
\State $A := \mathtt{CappedSumset}(A_L,A_S,t,\Delta)$
\State \Return $A$
\end{algorithmic}
\end{algorithm}

In the following we analyze this algorithm, proving the following lemma.

\begin{lem}[Recursive Splitting] \label{lem:recursivesplitting}
  Given integers $t,\Delta > 0$ with $t \ge 8 \Delta$ and a set $X \subseteq [t]$ of size~$n$, we can compute a set $A$ that \sDtapx $\cS(X;t)$ in time
  \[ \Oh\big( \big(n + \Tminconv(t/\Delta)\big) \log^8(nt/\Delta) \big). \]
\end{lem}

\paragraph{Recursion Depth}
Denote by $t_i$ the target bound on the $i$-th level of recursion. Let us first check that $t_i$ is monotonically decreasing. Initially we assume $t_0 \ge 8\Delta$. 
On any level with $t_i \ge 8\Delta$, the new target bound satisfies $t_{i+1} = (1+\eta)t_i/2 + \Delta \le \frac 34 t_i + \Delta < t_i$, where we used our choice of $\eta \le 1/2$.
Since we also have $t_{i+1} \ge t_i/2$, at some point we reach $t_i \in [4\Delta,8\Delta]$. The small items on this level are bounded by $t_i/k \le \Delta$, since $k \ge 8$. Hence, on the next level we will apply $\mathtt{Greedy}$ and the recursion stops. In particular, $t_i$ is monotonically decreasing throughout.

Note that 
\[ t_i = \Big(\frac{1+\eta}2\Big)^i t + \sum_{0 \le j < i} \Big(\frac{1+\eta}2\Big)^j \Delta. \]
Using $\sum_{0 \le j \le i} q^j \le \sum_{j \ge 0} q^j = 1/(1-q)$ yields
\[ t_i \le \Big(\frac{1+\eta}2\Big)^i t + \frac 2{1-\eta} \Delta \le \Big(\frac{1+\eta}2\Big)^i t + 4 \Delta, \]
where we used our choice of $\eta = 1/(2\log(t/\Delta)) \le 1/2$. Note that for $0 \le i \le \log(t/\Delta)$ we have $(1+\eta)^i \le \exp(\eta i) \le \exp(1/2) < 2$. Hence, for any $0 \le i \le \log(t/\Delta)$ the target bound satisfies
\begin{align} \label{eq:myeqtwo}
  t_i \le \frac{2t}{2^i} + 4 \Delta.
\end{align}
It follows that $t_{\log(t/\Delta)-1} \le 8 \Delta$, so the above argument shows that the recursion stops at the latest on level $\log(t/\Delta)$. We have therefore shown that the recursion depth of $\mathtt{RecursiveSplitting}$ is at most $\log(t/\Delta)$. In particular, inequality (\ref{eq:myeqtwo}) is applicable in each recursive call.

\paragraph{Correctness}
We inductively prove that with high probability for any recursive call of method $\mathtt{RecursiveSplitting}(X,t,\Delta)$ the output $A$ \sDtapx $\cS(X;t)$. Note that, as an output of $\mathtt{CappedSumset}$, $A$ is clearly \Dsparse, and thus we only need to show that $A$ \Dtapx $\cS(X;t)$, see Lemma~\ref{lem:correctness}. 
Since the recursion tree has total size $\Oh(t/\Delta)$, we can afford a union bound over all recursive calls. In particular, if we prove correctness of one recursive call with high probability, then the whole recursion tree is correct with high probability. Therefore, in the following we consider one recursive call.

\begin{lem} \label{lem:fromAtoS}
  With high probability, there exist $t_1,t_2 \in [(1+\eta)t/2, (1+\eta)t/2 + \Delta]$ such that $A$ \Dtapx $\cS(X_L;t) \plt \cS(X_1;t_1) \plt \cS(X_2;t_2)$.
\end{lem}
\begin{proof}
  We suppress the phrase ``with high probability'' throughout the proof.
  Inductively, $A_i$ \Dtprimeapx $\cS(X_i; t')$ for $t' = (1+\eta)t/2 + \Delta$ and any $i \in \{1,2\}$. By Lemma~\ref{lem:upshift}, there exists a value $t_i \in [t'-\Delta,t']$ such that $A_i$ \Dtapx $\cS(X_i; t') \cap [t_i] = \cS(X_i; t_i)$.
  Now Lemma~\ref{lem:cappedsumset} (Capped Sumset Computation) shows that $A_S$ \Dtapx $\cS(X_1;t_1) \plt \cS(X_2;t_2)$. Moreover, Lemma~\ref{lem:colorcoding} (Color Coding) yields that $A_L$ \Dtapx $\cS(X_L;t)$. It follows that the set $A$ \Dtapx $\cS(X_L;t) \plt \cS(X_1;t_1) \plt \cS(X_2;t_2)$.
\end{proof}

\begin{lem} \label{lem:usinghoeffding}
  Let $\Ssp(X_S;t)$ be the sparsification of $\cS(X_S;t)$ given by Lemma~\ref{lem:sparsification} and let $\bar{t} = (1+\eta)t/2$. 
  With high probability we have $\Ssp(X_S;t) \subseteq \cS(X_1;\bar{t}) \plt \cS(X_2;\bar{t})$.
\end{lem}
\begin{proof}
  For any $s \in \Ssp(X_S;t)$, fix a subset $Y \subseteq X_S$ with $\Sigma(Y) = s$ and write $Y = \{y_1,\ldots,y_\ell\}$. Let $Y_r := Y \cap X_r$ for $r \in \{1,2\}$. 
  Consider independent random variables $Z_1,\ldots,Z_\ell$ where $Z_i$ is uniformly distributed in $\{0,y_i\}$, and set $Z := Z_1+\ldots+Z_\ell$. Note that $Z$ has the same distribution as $\Sigma(Y_1)$ and $\Sigma(Y_2)$. Also note that $\Ex[Z] = \Sigma(Y)/2$. We use Hoeffding's inequality on $Z$ to obtain
  \[ \Pr\big[ Z - \Ex[Z] \ge \lambda \big] \le \exp\bigg( - \frac{2 \lambda^2}{\sum_i y_i^2} \bigg). \]
  Since $Y$ is a subset of the small items $X_S$, we have $y_i \le t/k$ for all $i$, and thus $\sum_i y_i^2 \le \sum_i y_i \cdot t/k \le t^2/k$. Setting $\lambda := \frac \eta 2 t$, we thus obtain
  \[ \Pr\big[ Z \ge \Ex[Z] + \frac \eta 2 t \big] \le \exp\bigg( - \frac{k \eta^2}{2} \bigg). 
  \]
  By our choice of $\eta := 1/(2\log(t/\Delta))$ and $k \ge C \log^3(nt/\Delta)$ we have $k \eta^2/2 \ge \frac C8 \log(nt/\Delta)$. Moreover, since $\Ex[Z] = \Sigma(Y)/2 \le t/2$, we obtain
  \[ \Pr\Big[ Z \ge (1+\eta) \tfrac t2 \Big] \le \Pr\big[ Z \ge \Ex[Z] + \tfrac \eta 2 t \big] \le \Big(\frac{\Delta}{tn}\Big)^{C/8}. \]
  For large $C$, this shows that with high probability $\Sigma(Y_1), \Sigma(Y_2) \le (1+\eta) \tfrac t2 = \bar{t}$, and hence $s = \Sigma(Y) \in \cS(X_1;\bar{t}) \plt \cS(X_2;\bar{t})$. Since $\Ssp(X_S;t)$ has size $\Oh(t/\Delta)$, we can afford a union bound over all $s \in \Ssp(X_S;t)$ to obtain that with high probability $\Ssp(X_S;t) \subseteq \cS(X_1;\bar{t}) \plt \cS(X_2;\bar{t})$.
\end{proof}

\begin{obs} \label{obs:partitionsubsetsum}
  For any partitioning $Z = Z_1 \cup Z_2$ we have $\cS(Z_1,t) \plt \cS(Z_2,t) = \cS(Z;t)$.
\end{obs}
\begin{proof}
  Follows from the fact that any subset sum of $Z$ can be uniquely written as a sum of a subset sum of $Z_1$ and a subset sum of $Z_2$. 
\end{proof}

\begin{lem} \label{lem:correctness}
  $A$ \wDtapx $\cS(X;t)$.
\end{lem}
\begin{proof}
  Let $t_1,t_2$ be as in Lemma~\ref{lem:fromAtoS}, in particular $\bar{t} = (1+\eta)t/2 \le t_1,t_2 \le (1+\eta)t/2 + \Delta \le t$.
  Using these bounds, Lemma~\ref{lem:usinghoeffding}, and Observation~\ref{obs:partitionsubsetsum}, we obtain with high probability
  \[ \Ssp(X_S;t) \subseteq \cS(X_1;\bar{t}) \plt \cS(X_2;\bar{t})
  \subseteq \cS(X_1;t_1) \plt \cS(X_2;t_2) \subseteq \cS(X_1;t) \plt \cS(X_2;t) = \cS(X_S;t). \]
  Since $\Ssp(X_S;t)$ \Dtapx $\cS(X_S;t)$, it now follows from Lemma~\ref{lem:sandwich} that with high probability the sumset $\cS(X_1;t_1) \plt \cS(X_2;t_2)$ \Dtapx $\cS(X_S;t)$. 
  
  Using Lemma~\ref{lem:apxsumset} (Sumset Property) and Observation~\ref{obs:partitionsubsetsum}, we obtain that with high probability
  \[ \cS(X_L;t) \plt \big(\cS(X_1;t_1) \plt \cS(X_2;t_2)\big) \;\; \text{\Dtapx}\;\; \cS(X_L;t) \plt \cS(X_S;t) = \cS(X;t). \]
  Lemma~\ref{lem:fromAtoS} and transitivity now imply that with high probability $A$ \Dtapx $\cS(X;t)$. It is easy to see that the inclusion $A \subseteq \cS(X;t)$ holds deterministically (i.e., with probability 1), and thus we even have that $A$ \wDtapx $\cS(X;t)$.
\end{proof}

This finishes the proof of correctness.

\paragraph{Running Time}
Lines 1-4 of Algorithm~\ref{alg:finalalgo} take time $\Oh(|X|)$, which sums to $\Oh(n)$ on each level of recursion, or $\Oh(n \log(t/\Delta))$ overall. 
On the $i$-th level of recursion, calling $\mathtt{ColorCoding}$ takes time 
\[ \Oh((|X_L| + k^2 \cdot \Tminconv(t_i/\Delta)) \log (n t/\Delta)), \] 
and calling $\mathtt{CappedSumset}$ is dominated by this running time. Since every item is large in exactly one recursive call, the terms $|X_L|$ simply sum up to $n$. 
For the remainder, for any $0 \le i \le \log(t/\Delta)$ we have $2^i$ instances, each with target bound $t_i \le 2t/2^i + 4\Delta$. Note that $2^i \cdot t_i \le 6t$, since $i \le \log(t/\Delta)$. 
Hence, by Lemma~\ref{lem:combineminconv} below, we can solve $2^i$ \minconv instances, each of size at most $t_i/\Delta$, in total time $\Oh(\Tminconv(t/\Delta))$. 
We can thus bound the time $\Oh(k^2 \cdot \Tminconv(t_i/\Delta) \log (n t/\Delta))$ summed over all recursive calls on level $i$ by $\Oh(k^2 \cdot \Tminconv(t/\Delta) \log (n t/\Delta))$. Over all levels, there is an additional factor $\log(t/\Delta)$. It follows that the total running time is
\[ \Oh\big( \big(n + k^2 \log(t/\Delta) \cdot \Tminconv(t/\Delta)\big) \log(nt/\Delta) \big). \]
Plugging in $k = \Oh(\log^3(nt/\Delta))$ yields a running time of $\Oh( (n + \Tminconv(t/\Delta)) \log^8(nt/\Delta) )$. It remains to prove the following lemma.

\begin{lem} \label{lem:combineminconv}
  There is an algorithm that solves $m$ given \minconv instances of sizes $n_1,\ldots,n_m$ in total time $\Oh(\Tminconv(\sum_i n_i))$.
\end{lem}

Note that for any ``natural'' time bound, say $\Tminconv(n) = \Theta(n^c)$, the lemma statement is immediate. Indeed, we must have $c \ge 1$, since writing the output requires time $\Theta(n)$. This implies $n_1^c + \ldots + n_m^c \le (n_1 + \ldots + n_m)^c$. It follows that we can solve $m$ given \minconv instances of sizes $n_1,\ldots,n_m$ in total time $\sum_i \Tminconv(n_i) = \Oh(\Tminconv(\sum_i n_i))$. 
The lemma proves that the same statement also holds for ``unnatural'' time bounds, by a reduction that packs multiple instances into one.

\begin{proof}
  Given $A_1,B_1 \in \N^{n_1},\ldots,A_{m},B_{m} \in \N^{n_{m}}$ our goal is to compute $C_1 \in \N^{n_1},\ldots,C_m \in \N^{n_m}$ satisfying $C_r[k] = \min_{0 \le i \le k} A_r[i] + B_r[k-i]$ for any $1 \le r \le m$ and $0 \le k < n_r$. We denote the largest input entry by $M$. 
  We assume the the instance sizes are sorted in non-increasing order $n_1 \ge \ldots \ge n_m$; note that it is easy to sort these numbers in time $\Oh(\sum_i n_i)$, which is dominated by the final running time. 
  We denote the prefix sums of the instance sizes by $s_r := \sum_{r'=1}^{r-1} n_{r'}$, and we set $s := s_{m+1}$.
  We construct sequences $A,B \in \N^{4s}$ by setting for any $1 \le r \le m$ and $0 \le i < n_r$:
  \begin{align*}
    A[2 s_r + i] := r^2 \cdot 2M + A_r[i], \qquad B[2 s_r + j] = r^2 \cdot 2M  + B_r[j],
  \end{align*}
  and setting all remaining entries to $\infty$ (we remark that a large finite number instead of $\infty$ is sufficient). 
  Then we compute $C = \minconv(A,B)$, that is, $C[k] = \min_{0 \le i \le k} A[i] + B[k-i]$ for any $0 \le k < 4s$.
  
  We claim that for any $1 \le r \le m$ and $0 \le k < n_r$ we have
  \begin{align} \label{eq:csr} 
    C[4s_r + k] = r^2 \cdot 4M + C_r[k]. 
  \end{align}
  The lemma follows from this claim, as we can infer $C_1,\ldots,C_m$ from $C$.
  To prove the claim, in one direction observe that
  \[ C[4s_r + k] \le \min_{0 \le i \le k} A[2s_r + i] + B[2s_r + k-i] = r^2 \cdot 4M + \min_{0 \le i \le k} A_r[i] + B_r[k-i] = r^2 \cdot 4M + C_r[k]. \]
  
  In the other direction, by construction we have $C[4s_r + k] = A[2s_x+i] + B[2s_y+j]$ for some $1 \le x,y \le m$ and $0 \le i < n_x$, $0 \le j < n_y$ such that $2s_x+i + 2s_y+j = 4s_r+k$, since all other entries of $A$ and $B$ are $\infty$. 
  We extend the notation $s_z$ to real numbers $z \in [1,m+1)$ by interpolating linearly, that is, $s_z := (\lfloor z \rfloor + 1 - z) \cdot s_{\lfloor z \rfloor} + ( z - \lfloor z \rfloor ) \cdot s_{\lfloor z \rfloor + 1}$.
  Since $n_1 \ge \ldots \ge n_m$ are sorted, the function $s_z$ is concave, in other words we have $(s_z + s_{z'})/2 \le s_{(z+z')/2}$ for any $z,z' \in [1,m+1)$. Hence,
  \[ 4 s_{(x+y+1)/2} \ge 2s_{x+1/2} + 2s_{y+1/2} = 2s_x + n_x + 2s_y + n_y > 2s_x+i + 2s_y+j = 4s_r+k \ge 4s_r. \]
  By monotonicity, this yields $(x+y+1)/2 > r$, resulting in $x+y > 2r-1$. Since $x,y,r$ are integers, we obtain $x+y \ge 2r$. We now consider three cases.
  
  \emph{Case 1:} For $x+y \ge 2r+1$, we have 
  \begin{align*} C[4s_r + k] = A[2s_x+i] + B[2s_y+j] \ge (x^2 + y^2) \cdot 2M &\ge (x^2 + (2r+1-x)^2) \cdot 2M  \\ &\ge \Big( \frac{2r+1}2 \Big)^2 \cdot 4M > (r^2 + 1) \cdot 4M \\&> r^2 \cdot 4M + C_r[k]. \end{align*}
  
  \emph{Case 2:} Similarly, for $x+y = 2r$ and $x \ne y$ we have
  \begin{align*} C[4s_r + k] = A[2s_x+i] + B[2s_y+j] \ge (x^2 + (2r-x)^2) \cdot 2M &\ge \big( (r-1)^2 + (r+1)^2 \big) \cdot 2M \\&= (r^2 + 1) \cdot 4M > r^2 \cdot 4M + C_r[k]. \end{align*}
  
  \emph{Case 3:} In the remaining case $x=y=r$, we have $i+j=k$ and
  \[ C[4s_r+k] = A[2s_r+i] + B[2s_r+j] = r^2 \cdot 4M + A_r[i] + B_r[j] \ge r^2 \cdot 4M + C_r[k]. \]
  In all cases we have $C[4s_r+k] \ge r^2 \cdot 4M + C_r[k]$. Together with the first direction, this proves the claim and thus the lemma.
\end{proof}

\subsubsection{Finishing the Proof}

We show how to use $\mathtt{RecursiveSplitting}$ to obtain an approximation scheme for \SubsetSum in time $\tOh(n + \Tminconv(1/\eps))$. Note that this proves Theorem~\ref{thm:mainalgo} as well as Corollary~\ref{cor:algo}.

Given $X,t$ and $\eps > 0$, let $\OPT := \max(\cS(X;t))$. 
Set $\Delta := \min\{\eps t, t/8\}$ and call the procedure $\mathtt{RecursiveSplitting}(X,t,\Delta)$ to obtain a set $A$ that \wDtapx $\cS(X;t)$. 
\begin{claim}
\label{claim:finish}
  With high probability, we have $\max(A) \ge \min\{\OPT, (1-\eps)t\}$. 
\end{claim}
\begin{proof}
  Consider $\Apxp(\OPT,A)$ and $\Apxm(\OPT,A)$. Since $\cS(X;t)$ does not contain any numbers in $(\OPT,t]$, and $A \subseteq \cS(X;t)$, we have $\Apxp(\OPT,A) \in \{\OPT,t+1\}$. If $\Apxp(\OPT,A) = \OPT$, then $A$ contains $\OPT$, so $\max(A) \ge \OPT$. 
Otherwise, if $\Apxp(\OPT,A) = t+1$, then $\Apxm(\OPT,A) \ge \Apxp(\OPT,A) - \Delta > t - \eps t$. In particular, $\max(A) \ge (1-\eps)t$. 
\end{proof}

We have thus shown how to compute a subset sum $\max(A)$ with $\max(A) \ge \min\{\OPT, (1-\eps)t\}$. It remains to determine a subset $Y \subseteq X$ summing to $\max(A)$. 
To this end, we retrace the steps of the algorithm, using the following idea. If $a \in \mathtt{CappedSumset}(A_1,A_2,t,\Delta)$, then $a \in A_1+A_2$, and thus we can simply iterate over all $a_1 \in A_1$ and check whether $a-a_1 \in A_2$, to reconstruct a pair $a_1 \in A_1, a_2 \in A_2$ with $a = a_1+a_2$ in linear time. 
Starting with $\max(A)$, we perform this trick in each recursive call of the algorithm, to reconstruct a subset summing to $\max(A)$.

The total running time of this algorithm is $\Oh( (n + \Tminconv(1/\eps)) \log^8(n/\eps) )$.

\input{partition_new}

\section{Open Problems}

A natural first question is to derandomize our reduction and the resulting approximation scheme for \SubsetSum. Another question is to prove hardness of \SubsetSum for other values of $\eps$. For example, when $\eps$ is much smaller than $1/n$, is the $\tOh(n/\eps)$-time algorithm conditionally (almost) optimal, or is there room for improvement? Lastly, it would be interesting to understand to what extend our ideas and the arguments of~\cite{MuchaW019}, most notably \cite[Section~5.3]{MuchaW019}, can be brought together, yielding further improvements for approximate \Partition. It is worth noting that gluing them together is a non-trivial task, mostly due to the fact that our algorithm does not seem to benefit from a distinction between ``large'' and ``small'' items, as in the algorithm of~\cite{MuchaW019}. Nevertheless, getting our hands on the additive structure of the problem seems the correct way to proceed, and incorporating the ideas in~\cite{MuchaW019} is a potential first step towards this goal.

\bibliographystyle{plain}
\bibliography{subsetsum}

\appendix

\section{Problem Variants of SubsetSum}
\label{sec:discussionproblems}

Recall that an approximation scheme for \SubsetSum asks to compute a subset $Y \subseteq X$ with $(1-\eps) \OPT \le \Sigma(Y) \le \OPT$, where $\OPT = \max\{ \Sigma(Y) \mid Y \subseteq X,\, \Sigma(Y) \le t \}$. 
Beyond this standard approximation goal, one can define many different variants of approximating \SubsetSum. 
For instance, computing such a subset $Y$ is not necessarily equivalent to just computing a number~$R$ with $(1-\eps)\OPT \le R \le \OPT$. 
To avoid these details in the problem definition, we consider the following two variants, which are in some sense the hardest and the simplest possible variants (subject to the strict constraint $\Sigma(Y) \le t$):
\begin{itemize}
  \item \ApxSubsetSum: Given $X,t$ and $\eps>0$, return any subset $Y \subseteq X$ satisfying $\Sigma(Y) \le t$ and $\Sigma(Y) \ge \min\{\OPT, (1-\eps)t\}$.
  \item \DisSubsetSum: Given $X,t$ and $\eps>0$, distinguish whether $\OPT = t$ or $\OPT < (1-\eps)t$. If $\OPT \in [(1-\eps)t,t)$ the output can be arbitrary.
\end{itemize}
Note that in Section~\ref{sec:algorithm} we solved the harder problem variant \ApxSubsetSum by a reduction to \minconv, while in Section~\ref{sec:lowerbound} we reduced \minconv to the simpler problem \DisSubsetSum. Also note that any algorithm for \ApxSubsetSum also solves \DisSubsetSum, since if $\OPT=t$ then the algorithm returns a set with sum in $[(1-\eps)t,t]$, while if $\OPT < (1-\eps)t$ then the algorithm returns a set with sum $\OPT < (1-\eps)t$.
Therefore, we obtain a full equivalence of \ApxSubsetSum, \DisSubsetSum, and \minconv. 


Intuitively, \ApxSubsetSum is the hardest and \DisSubsetSum is the simplest variant of approximating \SubsetSum. There are several further variants that are intermediate between \ApxSubsetSum and \DisSubsetSum, in the sense that any algorithm for \ApxSubsetSum also solves the intermediate variant and any algorithm for the intermediate variant also solves \DisSubsetSum. Since we prove \ApxSubsetSum and \DisSubsetSum to be equivalent, all intermediate variants are also equivalent. Examples of intermediate problem variants are as follows (we note that some of the reductions among these problem variants change $\eps$ by a constant factor):
\begin{itemize}
  \item Return any value in $[(1-\eps)\OPT,\OPT]$.
  \item Return any subset $Y \subseteq X$ with $(1-\eps) \OPT \le \Sigma(Y) \le t$.
  \item If $\OPT = t$, compute a subset $Y \subseteq X$ with $(1-\eps)t \le \Sigma(Y) \le t$, otherwise the output can be arbitrary.
  \item Distinguish whether $\OPT \ge (1-\eps/2) t$ or $\OPT < (1-\eps)t$. If $\OPT \in [(1-\eps)t,(1-\eps/2)t)$ the output can be arbitrary.
\end{itemize}

\end{document}

%% file: partition_new.tex

\section{Approximation Scheme for Partition}
\label{sec:apxpartition}

The goal of this section is to prove Theorem~\ref{thm:partition} and thus Corollary~\ref{cor:partition}. We restate Theorem~\ref{thm:partition} here for convenience (slightly reformulated).

\begin{thm}
\label{thm:partition_reduction}
For any $1 \le L \le 1/\eps$, we can approximate \Partition in time 
\[	\Oh\left(n + \left(\frac L \eps + L \cdot \Tminconv\left(\frac{1}{L \cdot \eps}\right) \right) \cdot \log^2\left(\frac n\eps\right) \right).	\]
\end{thm}


Recall that $\cS(X)$ is the set of all subset sums of $X$. On a high level, for $X = \{x_1,\ldots,x_n\}$ we compute $\cS(X)$ by evaluating the sumset $\cS(X) = \{0,x_1\} + \ldots + \{0,x_n\}$ in a tree-like fashion. In the bottom half of this tree, we approximate sumsets by Lemma~\ref{lem:unbsumsetcomp} from our \SubsetSum approximation scheme, see Lemma~\ref{lem:partition_Dapxsumset}. In the top half, where it remains to combine $L$ intermediate sumsets, we first round all sums computed so far to multiples of $\eps \cdot \Sigma(X) / L$; this is valid by Lemma~\ref{lem:weakrounding}. Then we combine the rounded sets by using \emph{exact} sumset computation via Fast Fourier Transform, see Lemma~\ref{lem:FFTsumset}. 

We start by formalizing these steps in the following lemmas, and then we combine them to obtain our approximation scheme.

\begin{lem}[Bottom Half] \label{lem:partition_Dapxsumset}
  Given a set $X\subseteq \mathbb{N}$ of size $n$ and $\Delta \in \mathbb{N}$, we can compute in time $\Oh(n + \Tminconv(\Sigma(X)/\Delta) \log n)$ a set $Z$ that \sDinftyapx $\cS(X)$.
\end{lem}
\begin{proof}
  Let $X = \{x_1,\ldots,x_n\}$ and note that $\cS(X) = \{0,x_1\} +\ldots+\{0,x_n\}$. We compute this sumset in a tree-like fashion. 
  On level $\ell$ of this tree, we have a partitioning $X_1 \cup \ldots \cup X_{n/2^\ell}$ of~$X$ and we have already computed sets $Z_1,\ldots,Z_{n/2^\ell}$ such that $Z_i$ \sDinftyapx $\cS(X_i)$. This holds on level $\ell=0$ with $X_i = \{x_i\}$ and $Z_i = \{0,x_i\}$. Now we argue how to go from level $\ell$ to level $\ell+1$. On level $\ell+1$ the partitioning will consist of the sets $X'_i := X_{2i-1} \cup X_{2i}$. Accordingly, we invoke Lemma~\ref{lem:unbsumsetcomp} to compute a set $Z'_i$ that \sDinftyapx $Z_{2i-1}+Z_{2i}$. By transitivity (Lemma~\ref{lem:transitivity}), $Z'_i$ also \sDinftyapx $\cS(X_{2i-1}) + \cS(X_{2i}) = \cS(X'_i)$. At level $\ell = \log n$ we then obtain a set~$Z$ that \sDinftyapx $\cS(X)$. 
  Since $Z_i \subseteq \cS(X_i) \subseteq [\Sigma(X_i)]$, all calls to Lemma~\ref{lem:unbsumsetcomp} on level~$\ell$ in total take time $\sum_{i=1}^{n/2^\ell} \Tminconv( \Sigma(X_i) / \Delta) = \Oh(\Tminconv(\Sigma(X)/\Delta))$, by Lemma~\ref{lem:combineminconv}. Over $\log n$ levels this yields the claimed running time.
\end{proof}

\begin{lem}[Top Half] \label{lem:FFTsumset}
  Given sets $Z_1,\ldots,Z_L \subseteq \mathbb{N}$, we can compute their sumset $Z_1+\ldots+Z_L$ in time $\Oh(\sigma \log \sigma \log L)$, where $\sigma := \max(Z_1) + \ldots + \max(Z_L)$. 
\end{lem}
\begin{proof}
  Using Fast Fourier Transform, we can compute a sumset $A+B \subseteq [U]$ in time $\Oh(U \log U)$. In particular, we can compute $Z_1 + Z_2$ in time 
  \[ \Oh\big((\max(Z_1)+\max(Z_2)) \log (\max(Z_1) + \max(Z_2))\big) = \Oh\big((\max(Z_1) + \max(Z_2)) \log \sigma\big). \] 
  Thus, we can compute all pairwise sumsets $Z_1+Z_2,\, Z_3+Z_4,\, Z_5+Z_6, \ldots$ in total time $\Oh(\sigma \log \sigma)$. This reduces the problem to an instance of size $L/2$. Using the same algorithm recursively, after $\log L$ such rounds we have computed the sumset $Z_1+\ldots+Z_L$. This runs in total time $\Oh(\sigma \log \sigma \log L)$. 
  
  Note that the value of $\sigma$ does not change throughout this algorithm, since $\max(Z_1+Z_2)+\max(Z_3+Z_4)+\ldots = \max(Z_1) + \max(Z_2) + \max(Z_3) + \max(Z_4) + \ldots$.
\end{proof}

\begin{lem}[Weak Rounding] \label{lem:weakrounding}
  For any set $Z \subseteq \mathbb{N}$ and $R \in \mathbb{N}$ we write $\lfloor Z/R \rfloor := \{ \lfloor z/R \rfloor \mid z \in Z\}$ and $R \cdot Z := \{ R\cdot z \mid z \in Z\}$.
  
  Let $Z_1,\ldots,Z_L \subseteq \mathbb{N}$ and let $R \in \mathbb{N}$. For any sum $s$ in $Z_1 + \ldots + Z_L$ and the corresponding sum $s'$ in $R \cdot (\lfloor Z_1/R \rfloor + \ldots + \lfloor Z_L / R \rfloor)$ we have $s \ge s' \ge s - L\cdot R$.
\end{lem}
\begin{proof}
  Consider any sum $s = s_1+\ldots+s_L \in Z_1+\ldots+Z_L$ and the corresponding sum $s' = R \cdot (\lfloor s_1/R \rfloor + \ldots + \lfloor s_L/R \rfloor )$. Since we round down, we have $s' \le s$. Moreover, we have 
  \[ s' \ge R \cdot ((s_1/R - 1) + \ldots + (s_L/R - 1)) = s - L\cdot R. \qedhere \]
\end{proof}

We are now ready to describe our approximation scheme for \Partition. Given a set $X \subseteq \mathbb{N}$ of size $n$, let $\OPT := \max\{ \Sigma(Y) \mid Y \subseteq X,\, \Sigma(Y) \le \sigma/2\}$ and set $\sigma := \Sigma(X)$ and $\Delta := \eps \cdot \sigma / 4$.
We proceed as follows.

\begin{enumerate}
\item We partition $X$ into $X_1 \cup \ldots \cup X_L$ such that each $X_i$ satisfies (1) $\Sigma(X_i) \le 4\sigma/L$ or (2) $|X_i| = 1$. To this end, items larger than $2\sigma/L$ get assigned to singleton parts; this uses at most $L/2$ parts. We partition the remaining items greedily by adding items to $X_i$ until its sum is more than $2\sigma/L$; this results in at most $L/2$ non-empty parts, each with a sum in $[2\sigma/L,4\sigma/L]$. 

\item Next, we compute sets $Z_1,\ldots,Z_L$ such that $Z_i$ \sDinftyapx $\cS(X_i)$. For parts~$X_i$ with $|X_i| = 1$ we can simply set $Z_i := \{0\} \cup X_i = \cS(X_i)$. For parts $X_i$ with $\Sigma(X_i) \le 4\sigma/L$, Lemma~\ref{lem:partition_Dapxsumset} computes the required set $Z_i$ in time 
\[ \Oh\Big(|X_i| + \Tminconv\Big(\frac{\Sigma(X_i)}\Delta\Big) \log n\Big) = \Oh\Big(|X_i| + \Tminconv\Big(\frac 1{L \cdot \eps}\Big) \log n\Big). \] 
Over all $1 \le i \le L$, the total running time of this step is $\Oh(n + L \cdot \Tminconv(1/(L \cdot \eps)) \log n)$.

\item Next, we round every $Z_i$ to $\tilde Z_i := \lfloor Z_i / R \rfloor$ for $R := \Delta/L$. We compute $\tilde Z_1 + \ldots + \tilde Z_L$ using Lemma~\ref{lem:FFTsumset}. We multiply each number in the resulting set by $R$, to obtain $S = R \cdot( \lfloor Z_1 / R \rfloor + \ldots + \lfloor Z_L / R \rfloor )$. Note that $\max(\tilde Z_1)+\ldots+\max(\tilde Z_L) \le \sigma/R$, so this step runs in time $\Oh(\frac \sigma R \log \frac \sigma R \log L) = \Oh(\frac L \eps \log \frac L \eps \log L)$.

\item Finally, pick the largest number in $S$ that is at most $\sigma/2$. Retrace the previous steps to construct a corresponding subset $Y \subseteq X$. If $\Sigma(Y) \le \sigma/2$ then return $Y' := Y$, otherwise return $Y' := X \setminus Y$.
\end{enumerate}

\paragraph*{Correctness:}
Let $Y^* \subseteq X$ be a subset summing to $\OPT$. 
Since $Z_i$ \Dinftyapx $\cS(X_i)$, by Lemma~\ref{lem:apxsumset} $Z_1+\ldots+Z_L$ \Dinftyapx $\cS(X_1) + \ldots + \cS(X_L) = \cS(X)$. Hence, there is a sum in $Z_1+\ldots+Z_L \cap [\OPT-\Delta,\OPT]$. By Lemma~\ref{lem:weakrounding}, the corresponding sum in $S$ lies in $[\OPT-2\Delta, \OPT]$. In particular, the largest number in $S$ that is at most $\sigma/2$ is at least $\OPT-2\Delta$.

Now pick any number in $S \cap [\OPT-2\Delta, \sigma/2]$.
This corresponds to a sum in $Z_1+\ldots+Z_k$ lying in the range $[\OPT-2\Delta, \sigma/2 + \Delta]$, by Lemma~\ref{lem:weakrounding}. In particular, it corresponds to a subset $Y \subseteq X$ summing to a number in $[\OPT-2\Delta, \sigma/2 + \Delta]$. By possibly replacing $Y$ by $X \setminus Y$, we obtain a subset $Y'$ of $X$ summing to a number in $[\OPT-2\Delta, \sigma/2]$. Since any subset of $X$ summing to at most $\sigma/2$ sums to at most $\OPT$, we obtain $\Sigma(Y') \in [\OPT-2\Delta, \OPT]$. Since we set $\Delta := \eps \sigma / 4$, we obtain the desired  $(1-\eps)\OPT \le \Sigma(Y') \le \OPT$ if we assume $\OPT \ge \sigma/4$. 

It remains to argue why we can assume $\OPT \ge \sigma/4$. Write $X = \{x_1,\ldots,x_n\}$ and suppose that $x_n$ is the largest item in $X$.
If $x_n > \sigma/2$, then $x_n$ cannot be picked by any subset summing to at most $\sigma/2$, which yields $\sigma/2 > \sigma - x_n = x_1+\ldots+x_{n-1} = \OPT$, so we can solve the problem in linear time $\Oh(n)$. Otherwise, all items are bounded by at most $\sigma/2$, so there exists $1\le i^* \le n$ such that the prefix sum $\sum_{i=1}^{i^*} x_i$ lies in $[\sigma/4, 3\sigma/4]$. It follows that one of $\sum_{i=1}^{i^*} x_i$ or $\sum_{i=i^*+1}^{n} x_i$ lie in $[\sigma/4, \sigma/2]$, which shows $\OPT \ge \sigma/4$. This justifies our assumption $\OPT \ge \sigma/4$ from above and thus finishes the correctness argument.

\paragraph*{Running Time:} 
To analyze the running time of the retracing in step 4, we use the following: After we have computed a sumset $A+B \subseteq [U]$ in time $\Oh(U \log U)$, given a number $x \in A+B$ we can find a pair $a \in A, b \in B$ with $a+b=x$ in time $\Oh(U)$, by simply iterating over all $a \in A$ and checking whether $x-a \in B$. Similarly, after we have computed $C = \minconv(A,B)$ for vectors $A,B \in \mathbb{N}^U$, given an output index $k$ we can find a pair $i+j=k$ such that $A[i]+B[j] = C[k]$ in time $\Oh(U)$, again by simple brute force. 
Therefore, if we perform all previous steps backwards, we can reconstruct a corresponding subset by spending asymptotically at most as much time as these steps took originally.

The total time is thus dominated by steps 2 and 3,
\[ \Oh\left(n + L \cdot \Tminconv\left(\frac 1{L \cdot \eps}\right) \log n + \frac L \eps \log \frac L \eps \log L\right). \]
We can roughly bound this by 
\[ \Oh\left(n + \left(\frac L\eps + L \cdot \Tminconv\left(\frac 1{L \cdot \eps}\right)\right) \cdot \log^2\left(\frac n\eps\right)\right), \]
which is the bound claimed in Theorem~\ref{thm:partition_reduction}. This finishes the proof.